\documentclass[letterpaper,times]{IONconf-v2}
\usepackage{amsmath}
\usepackage{amssymb}
\usepackage{amsthm}
\usepackage{bm}
\usepackage{tikz}
\usetikzlibrary{arrows.meta, calc, decorations.pathreplacing, fit, positioning}

\usepackage[hidelinks]{hyperref}

\theoremstyle{plain}
\newtheorem{theorem}{Theorem}
\newtheorem{lemma}{Lemma}
\newtheorem{proposition}{Proposition}
\newtheorem{corollary}{Corollary}
 
\theoremstyle{remark}

\newcommand{\ones}{\mathbf{1}}
\renewcommand{\vec}[1]{{\boldsymbol{\mathrm{#1}}}}

\articletype{Original paper}%

\received{TBD}
\revised{TBD}
\accepted{TBD}
\doi{10.33012/navi.XXX}
\journalname{NAVIGATION}
\journalvolume{TBD}
\journalnumber{TBD}

\title{Geometric Analysis of Doppler-Based Navigation with Low Earth Orbit Satellites}

\author[1]{Carlos Caravaca Gallego*}
\author[1]{Pini Gurfil}
\author[2]{Hector Rotstein}

\authormark{Caravaca Gallego \textit{et al.}}

\address[1]{The Stephen B.~Klein Faculty of Aerospace Engineering, Technion -- Israel Institute of Technology, Haifa, Israel}

\address[2]{Rafael Advanced Defense Systems / Department of Computer Science, Technion -- Israel Institute of Technology, Haifa, Israel}

\corres{*Carlos Caravaca Gallego, The Stephen B.~Klein Faculty of Aerospace Engineering, Technion -- Israel Institute of Technology, Haifa 3200003, Israel.\\ \email{carlosc@campus.technion.ac.il}}

\begin{document}

\abstract[Abstract]{The increasing vulnerability of Global Navigation Satellite Systems has motivated renewed interest in Doppler-based navigation using low Earth orbit satellites, which can determine position, velocity, clock bias, and clock drift from carrier Doppler measurements. However, the geometric dilution of precision (GDOP) in this eight-state problem behaves fundamentally differently from the GDOP in pseudorange-based navigation. In particular, volume-based satellite selection, effective for pseudorange-based GDOP minimization, has been empirically found to perform poorly for Doppler GDOP. This paper establishes a geometric foundation for Doppler GDOP characterization. A closed-form geometric parameterization of the Doppler measurement Jacobian in terms of elevation, azimuth, inclination, and altitude ratio is derived. It is shown that the clock bias sensitivity depends on elevation, altitude, and the angle between the satellite velocity vector and the line of sight. This sensitivity produces a correlation with the clock drift that no geometric arrangement can remove. A Schur complement decomposition of the eight-state information matrix is performed, yielding an exact GDOP inflation formula, governed by a collinearity coefficient that measures the alignment between the clock bias column and a seven-state subspace. It is proven that the geometric coupling that renders clock bias observable is the same coupling that inflates estimation covariance, and that satellite altitude diversity reduces the collinearity by separating satellites at equal elevation onto distinct sensitivity bands. An example employing a two-shell satellite configuration illustrates the analysis.}


\maketitle



\vspace{-15pt}

\section{Introduction}\label{sec:intro}

The increasing vulnerability of Global Navigation Satellite Systems (GNSS) to radio-frequency interference \parencite{EASA2023} has motivated renewed interest in navigation using Low Earth Orbit (LEO) satellite signals \parencite{Iannucci2020, Reid2018}. Large constellations deployed primarily for broadband communications, including Starlink, OneWeb, and Iridium, can be exploited opportunistically for positioning and timing. Compared to medium Earth orbit GNSS satellites, LEO satellites are far closer to the receiver, offering a free-space path-loss advantage that, at comparable transmit power and antenna gain, can make the received signal up to three orders of magnitude stronger. The proliferation of mega-constellations ensures that tens to hundreds of satellites are simultaneously visible from any point on Earth, providing geometric diversity across multiple orbital planes and altitudes \parencite{Iannucci2020, Reid2018}.

Among the approaches to LEO-based positioning, Doppler-only navigation is distinguished by its ability to determine all eight navigation states, i.e. three position components, receiver clock bias, three velocity components, and clock drift, from carrier Doppler measurements alone, without requiring pseudorange observables. \textcite{Psiaki2021} established the analytical framework for this eight-state formulation, introducing the structure of the measurement-equation Jacobian and the scaling parameters that render the Geometric Dilution of Precision (GDOP) analysis meaningful for the Doppler-based navigation problem. The key enabler is the high angular rate of LEO satellites, which produces large and rapidly-changing Doppler shifts that carry rich geometric information. Simulation results have demonstrated that position errors of 1--5~meters are achievable with large LEO constellations, performance comparable to pseudorange-based GNSS \parencite{Psiaki2021}.

However, tracking and processing signals from every visible satellite may be computationally prohibitive, particularly for resource-constrained receivers. Therefore, the satellite selection problem, i.e. choosing a subset of visible satellites that keeps good navigation performance while reducing computational burden, becomes important for LEO Doppler navigation. Efficient selection requires understanding how satellite geometry affects GDOP. In pseudorange-based navigation, minimizing GDOP is well approximated by maximizing the determinant of the measurement Jacobian Gram matrix, which is proportional to the volume spanned by the satellite geometry vectors \parencite{BlancoDelgado2017}. This volume-based heuristic is effective because the pseudorange clock bias column, a vector of ones, is independent of satellite geometry. Whether analogous proxies exist for Doppler GDOP depends on the structure of the eight-state Doppler measurement Jacobian.

\textcite{Moore2024} investigated this problem using real orbital data of the OneWeb constellation, comparing several satellite selection metrics for Doppler GDOP minimization. The volume of the convex shape formed by the LOS unit vectors showed no correlation with Doppler GDOP (Spearman rank correlation $r = -0.18$), whereas determinant-based metrics operating on the full Doppler Jacobian achieved strong correlation ($r = -0.91$). The analysis by \textcite{Moore2024} attributes the failure of volume-based selection to velocity diversity, which the volume-based metric ignores. What that empirical work does not provide is the analytical form of the measurement Jacobian, specifically the clock bias column and its influence on Doppler GDOP; without these, one cannot determine whether alternative geometric proxies exist or whether the full Jacobian computation is unavoidable. The current paper provides the missing analysis. \textcite{Psiaki2021} already identified the geometric diversity of the LOS rate vectors as important to Doppler GDOP and showed, through constellation node-clustering experiments, that poor rate-vector diversity degrades performance by a large factor; the present contribution is complementary, isolating the clock bias column as the structural source of the volume-heuristic failure and characterizing it in closed form.

Recent studies approach LEO Doppler-navigation geometry by combining Doppler and pseudorange DOP \parencite{McLemoreDoP2022} or through position-only orbital-plane information conditioning for Doppler-only positioning \parencite{Thornton2026}, but do not isolate the Doppler-only clock bias column or its contribution to GDOP through the seven-state subspace, which is the objective of the present analysis. Another relevant work is that of \textcite{Baron2024}, which addresses the same eight-state Doppler problem. \textcite{Baron2024} derived closed-form analytical expressions for all elements of a high-fidelity Doppler Jacobian, including the clock bias column, together with the same scaling-parameter nondimensionalization and a family of dilution-of-precision metrics. The distinction here is one of variables and of aim; namely, \textcite{Baron2024} differentiated the measurement model with respect to the Earth-centered position and velocity vectors, in a form suited to numerical evaluation along a trajectory, whereas the present parameterization is expressed instead in the topocentric variables seen by the receiver. This renders the clock bias column an explicit function of these variables and brings it within reach of the collinearity analysis performed in this work. The closed-form geometric parameterization is, therefore, the enabling step; the main contributions are the structural results it yields, none of which appeared in prior work: the sign-definiteness of the clock-bias sensitivity, the exact Schur-complement inflation identity, and the altitude-diversity mechanism.

We first show that the clock bias sensitivity depends on elevation, altitude, and the angle between the satellite velocity and the LOS, bounded between an in-plane and a cross-track envelope. We then identify the clock-bias coupling structure for all visible satellites, and prove that the sensitivity is negative for every satellite above the horizon, so that the clock bias column is irreducibly correlated with the clock drift column; the coupling operates through two mechanisms, one linking clock bias to position sensitivity through the shared LOS rate, and the other to the velocity $z$-component through the shared elevation factor. In addition, we perform a Schur-complement decomposition of the eight-state information matrix that expresses GDOP inflation exactly in terms of the collinearity coefficient between the clock bias column and the seven-state subspace. We prove that perfect decorrelation is impossible for any constellation, establishing that the GDOP cost of estimating clock bias is structurally unavoidable.

Finally, we identify altitude diversity as the primary decorrelation lever and the satellite heading as a weaker one, and characterize the sensitivity-decorrelation tradeoff that limits the effectiveness of altitude-based heuristics. An example with a synthetic two-shell configuration, i.e. satellites distributed across two distinct orbital altitudes, illustrates the analysis and verifies the inflation formula.

\vspace{-10pt}

\section{Background}\label{sec:background}
\vspace{-5pt}

This section reviews the Doppler measurement model and introduces the quantities used throughout the paper. The eight-state Jacobian formulation for Doppler-only navigation is presented first, followed by a summary of the pseudorange GDOP framework that motivates the comparison with the Doppler case. The section concludes by examining the clock bias column of the Doppler Jacobian, whose geometry dependence distinguishes it from its pseudorange counterpart and sets the stage for the geometric analysis in subsequent sections.

\vspace{-10pt}

\subsection{Doppler Measurement Model Jacobian}\label{subsec:doppler_model}

\textcite{Psiaki2021} showed that carrier Doppler-shift measurements from eight or more simultaneously-visible LEO satellites are sufficient for determining all eight navigation unknowns, i.e. three position components, receiver clock bias, three velocity components, and receiver clock drift, without requiring pseudorange measurements. The simplified measurement model for the $j$th satellite is
\begin{equation}\label{eq:simplified_doppler}
    -\lambda D^j = (\hat{\vec\rho}^j)^T (\vec{v}^j - \vec{v}) + c\frac{d\delta_R}{dT_R} - c\frac{d\delta^j}{dT^j}
\end{equation}
where $D^j$ is the measured Doppler shift, $\lambda$ is the carrier wavelength, $\hat{\vec\rho}^j$ is the unit line-of-sight (LOS) vector from the receiver to the satellite, $\vec{v}$ and $\vec{v}^j$ are the respective receiver and satellite velocity vectors, $c$ is the speed of light, and $\delta_R$, $\delta^j$ are the receiver and satellite clock biases, respectively. Note that \textcite{Psiaki2021} defines the LOS unit vector from the satellite to the receiver; the expressions in this paper are equivalent, rewritten for the receiver-to-satellite convention. The receiver is modeled as stationary on a non-rotating Earth ($\vec{v} = \mathbf{0}$); receiver motion, which is slow relative to the satellite orbital speed (about 7.5~km/s) and includes the Earth-rotation velocity of an Earth-fixed receiver, shifts the relative velocity that enters the measurement at the few-percent level; the effect of this shift on the clock-bias structure is quantified in Section~\ref{sec:parameterization}, after the relevant geometric quantities are defined. 

The eight-state navigation vector to be estimated is $\vec{x} = [\vec{r}^T,\; c\,\delta_R,\; \vec{v}^T,\; c\,\dot{\delta}_R]^T$, where $\vec r$ denotes the position vector and the clock terms are scaled by the speed of light to obtain units of length and velocity, respectively. The measurement model in Eq.~\eqref{eq:simplified_doppler} is expressed in coordinate-free form: it involves only inner products and magnitudes of physical vectors and is, therefore, independent of the reference frame in which those vectors are resolved. For the numerical evaluation of the Jacobian, the position, velocity, and acceleration vectors are resolved in an Earth-Centered Inertial (ECI) frame. For the geometric analysis of Section~\ref{sec:parameterization}, where the natural variables are the satellite elevation and azimuth as seen by the receiver, the LOS quantities are resolved in a local East--North--Up (ENU) topocentric frame centered at the receiver. The transformation between the two frames is a rotation determined by the receiver latitude and longitude and does not affect the scalar Jacobian entries. Linearization about a nominal trajectory yields the $N \times 8$ Jacobian matrix \parencite{Psiaki2021}
\begin{equation}\label{eq:jacobian}
A_{\mathrm{GDOP}} =
\begin{bmatrix}
    A_p & A_b & A_v & A_d
\end{bmatrix} =
\begin{bmatrix}
(\dot{\hat{\vec\rho}}^1)^T / \gamma
  & [(\hat{\vec\rho}^1)^T \dot{\bar{\vec{v}}}^1
    + (\dot{\hat{\vec\rho}}^1)^T \bar{\vec{v}}^1] / \eta
  & (\hat{\vec\rho}^1)^T & 1 \\
\vdots & \vdots & \vdots & \vdots \\
(\dot{\hat{\vec\rho}}^N)^T / \gamma
  & [(\hat{\vec\rho}^N)^T \dot{\bar{\vec{v}}}^N
    + (\dot{\hat{\vec\rho}}^N)^T \bar{\vec{v}}^N] / \eta
  & (\hat{\vec\rho}^N)^T & 1
\end{bmatrix}
\end{equation}
Here, $A_p \in \mathbb{R}^{N\times 3}$, $A_b \in \mathbb{R}^{N}$, $A_v \in \mathbb{R}^{N \times 3}$, and $A_d \in \mathbb{R}^{N}$ denote the column blocks associated with position, clock bias, velocity, and clock drift, respectively, and $a_p^j$, $a_b^j$, $a_v^j$, $a_d^j$ denote the $j$th rows of the corresponding blocks. The column signs follow \textcite{Psiaki2021}, whose LOS convention is satellite-to-receiver; in the convention adopted here the partial derivatives of the measurement of Eq.~\eqref{eq:simplified_doppler} with respect to position and velocity are the negatives of the corresponding columns, whereas the clock-bias and clock-drift columns are convention-invariant. Reversing the sign of any column conjugates the Gram matrix by a diagonal sign matrix and leaves every dilution and collinearity quantity in this paper unchanged, so the signs of the source are retained. The barred quantities $\bar{\vec{v}}^j \equiv \vec{v} - \vec{v}^j$ and $\dot{\bar{\vec{v}}}^j$ are the relative velocity and acceleration of the receiver with respect to satellite $j$; for the stationary receiver, $\bar{\vec{v}}^j = -\vec{v}^j$ and $\dot{\bar{\vec{v}}}^j = -\dot{\vec{v}}^j$, where $\dot{\vec{v}}^j$ is the satellite acceleration vector. Notice that $A_d = \ones$ is a constant column of ones, and that we have used the scaling parameters:
%
\begin{equation}\label{eq:gamma_eta}
\gamma = \frac{1}{1 - R_E/a_{orb}} \sqrt{\frac{\mu}{a_{orb}^3}}, \qquad
\eta = \frac{R_E/a_{orb}}{1 - R_E/a_{orb}} \cdot \frac{\mu}{a_{orb}^2}
\end{equation}
where $\mu$ is Earth's gravitational parameter, $a_{orb}$ is the satellite orbital radius, and $R_E$ is Earth's radius. The satellite orbits are assumed circular throughout, so that the orbital radius equals the semimajor axis. Operational LEO constellations are near-circular, so that the orbital radius, speed, and acceleration that enter the Jacobian depart from their circular values by an amount of order the orbital eccentricity, which is negligible for the present geometric analysis. These parameters nondimensionalize the Jacobian matrix so that all column magnitudes are of order unity. Notice that as the orbital altitude decreases, both $\gamma$ and $\eta$ tend to infinity, reflecting the stronger Doppler signatures available from LEO satellites \parencite{Psiaki2021}. Figure~\ref{fig:doppler_geometry} illustrates the measurement geometry for two satellites at different orbital altitudes, showing the LOS direction $\hat{\vec\rho}^j$, satellite velocity $\vec{v}^j$, LOS rate vector $\dot{\hat{\vec\rho}}^j$, and the range rate $\dot{\rho}^j$ that enter the Jacobian columns.

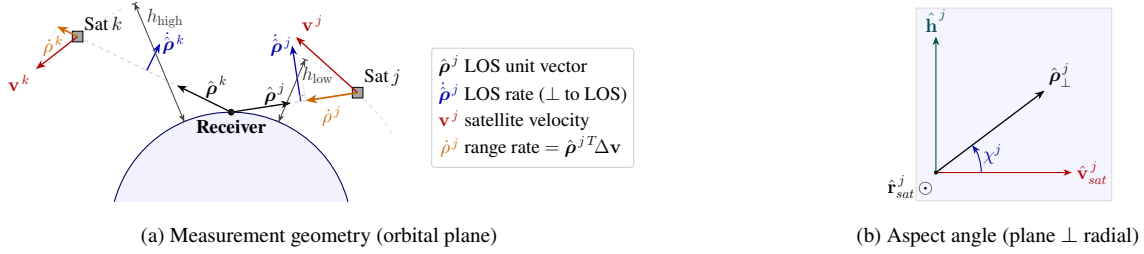
\begin{figure}[htbp!]
\centering
\begin{minipage}[b]{0.65\textwidth}
\centering
\scalebox{0.50}{%
\begin{tikzpicture}[
    >=Stealth,
    sat/.style={fill=gray!70, draw=black, line width=0.4pt,
                minimum size=8pt, inner sep=0pt, rectangle},
    vec/.style={->, thick},
    losline/.style={dashed, gray!40, thin},
    ann/.style={font=\Large},
  ]
  \def\RE{3.2}
  \def\hL{1.8}
  \def\hH{3.4}
  \def\aL{42}
  \def\aH{38}
  \pgfmathsetmacro{\posL}{90-\aL}
  \pgfmathsetmacro{\posH}{90+\aH}
  \def\vLenL{2.2}
  \def\vLenH{1.4}
  \def\lrLenL{1.5}
  \def\lrLenH{0.8}
  \fill[blue!5]
    (15:\RE) arc[start angle=15, end angle=165, radius=\RE]
    -- cycle;
  \draw[thick, blue!30!black]
    (15:\RE) arc[start angle=15, end angle=165, radius=\RE];
  \coordinate (R) at (90:\RE);
  \fill[black] (R) circle (2.5pt);
  \node[below=2pt, font=\Large\bfseries] at (R) {Receiver};
  \coordinate (SL) at ({\posL}:{\RE+\hL});
  \coordinate (SH) at ({\posH}:{\RE+\hH});
  \draw[gray!30, thin, dashed]
    ({\posL-16}:{\RE+\hL})
    arc[start angle={\posL-16}, end angle={\posL+16},
        radius={\RE+\hL}];
  \draw[gray!30, thin, dashed]
    ({\posH-12}:{\RE+\hH})
    arc[start angle={\posH-12}, end angle={\posH+12},
        radius={\RE+\hH}];
  \draw[losline] (R) -- (SH);
  \node[sat] at (SH) {};
  \node[ann, above right=2pt] at (SH) {Sat\,$k$};
  \draw[vec, black, line width=1pt] (R) -- ($(R)!1.6cm!(SH)$)
    node[pos=0.5, above right=-4pt, ann] {$\hat{\vec\rho}^{\,k}$};
  \pgfmathsetmacro{\tanH}{\posH+90}
  \draw[vec, red!70!black, line width=1pt] (SH) -- ++({\tanH}:\vLenH)
    coordinate (vHend)
    node[ann, below left=1pt, red!70!black] {$\vec{v}^{\,k}$};
  \coordinate (lmH) at ($(R)!0.55!(SH)$);
  \draw[vec, blue!70!black, line width=1pt] (lmH) --
    ($(lmH)!\lrLenH cm!-90:(SH)$)
    node[pos=1.0, right=-2pt, ann, blue!70!black]
      {$\dot{\hat{\vec\rho}}^{\,k}$};
  \coordinate (pfH) at ($(R)!(vHend)!(SH)$);
  \draw[densely dotted, gray!50, thin] (vHend) -- (pfH);
  \draw[vec, orange!80!black, line width=1.2pt] (SH) -- (pfH)
    node[midway, below left=-1pt, ann, orange!80!black]
      {$\dot\rho^{\,k}$};
  \draw[losline] (R) -- (SL);
  \node[sat] at (SL) {};
  \node[ann, above right=2pt] at (SL) {Sat\,$j$};
  \draw[vec, black, line width=1pt] (R) -- ($(R)!1.6cm!(SL)$)
    node[pos=0.5, above right=-2pt, ann] {$\hat{\vec\rho}^{\,j}$};
  \pgfmathsetmacro{\tanL}{\posL+90}
  \draw[vec, red!70!black, line width=1pt] (SL) --
    ++({\tanL}:\vLenL) coordinate (vLend)
    node[ann, above right=1pt, red!70!black] {$\vec{v}^{\,j}$};
  \coordinate (lmL) at ($(R)!0.55!(SL)$);
  \draw[vec, blue!70!black, line width=1pt] (lmL) --
    ($(lmL)!\lrLenL cm!90:(SL)$)
    node[pos=1.0, right=-20pt, ann, blue!70!black]
      {$\dot{\hat{\vec\rho}}^{\,j}$};
  \coordinate (pfL) at ($(R)!(vLend)!(SL)$);
  \draw[densely dotted, gray!50, thin] (vLend) -- (pfL);
  \draw[vec, orange!80!black, line width=1.3pt] (SL) -- (pfL)
    node[midway, below=3pt, ann, orange!80!black]
      {$\dot\rho^{\,j}$};
  \pgfmathsetmacro{\sqsz}{0.15}
  \coordinate (sqA) at ($(pfL)!\sqsz cm!(R)$);
  \coordinate (sqB) at ($(pfL)!\sqsz cm!(vLend)$);
  \coordinate (sqC) at ($(sqA)+(sqB)-(pfL)$);
  \draw[thin, gray!60] (sqA) -- (sqC) -- (sqB);
  \pgfmathsetmacro{\hAL}{90 - \aL*0.55}
  \draw[thin, |<->|, gray!50!black]
    ({\hAL}:\RE) -- ({\hAL}:{\RE+\hL})
    node[near end, right=-1pt, ann, gray!40!black] {$h_{\mathrm{low}}$};
  \pgfmathsetmacro{\hAH}{90 + \aH*0.6}
  \draw[thin, |<->|, gray!50!black]
    ({\hAH}:\RE) -- ({\hAH}:{\RE+\hH})
    node[very near end, right=1pt, ann, gray!40!black] {$h_{\mathrm{high}}$};
  \node[anchor=north west, inner sep=5pt, font=\Large,
        fill=white, draw=gray!40, rounded corners=2pt]
    at (5.3, 4.9) {%
    \begin{tabular}{@{}r@{\;}l@{}}
      $\hat{\vec\rho}^{\,j}$
        & LOS unit vector \\[2pt]
      \textcolor{blue!70!black}{$\dot{\hat{\vec\rho}}^{\,j}$}
        & LOS rate ($\perp$ to LOS) \\[2pt]
      \textcolor{red!70!black}{$\vec{v}^{\,j}$}
        & satellite velocity \\[2pt]
      \textcolor{orange!80!black}{$\dot\rho^{\,j}$}
        & range rate $= \hat{\vec\rho}^{\,j\,T}\!\Delta\vec{v}$
    \end{tabular}};
\end{tikzpicture}%
}
\\[3pt]{\footnotesize (a) Measurement geometry (orbital plane)}
\end{minipage}\hfill
\begin{minipage}[b]{0.34\textwidth}
\centering
\scalebox{0.50}{%
\begin{tikzpicture}[>=Stealth, scale=1.5,
    ax/.style={->, thick},
    ann/.style={font=\Large}]
  \def\L{2.4}
  \pgfmathsetmacro{\chideg}{37}
  \fill[blue!4] (-0.35,-0.5) rectangle (\L+0.7,\L+0.5);
  \draw[gray!30] (-0.35,-0.5) rectangle (\L+0.7,\L+0.5);
  \draw[ax, red!70!black] (0,0) -- (\L,0)
    node[right, ann, red!70!black] {$\hat{\vec v}_{sat}^{\,j}$};
  \draw[ax, teal!60!black] (0,0) -- (0,\L)
    node[above, ann, teal!60!black] {$\hat{\vec h}^{\,j}$};
  \draw[ax, black, line width=1pt]
    (0,0) -- ({\L*cos(\chideg)},{\L*sin(\chideg)})
    node[above right=-1pt, ann] {$\hat{\vec\rho}_\perp^{\,j}$};
  \draw[->, blue!60!black, line width=0.8pt]
    (0.8,0) arc[start angle=0, end angle=\chideg, radius=0.8];
  \node[ann, blue!60!black]
    at ({1.05*cos(\chideg/2)},{1.05*sin(\chideg/2)}) {$\chi^j$};
  \fill (0,0) circle (1pt);
  \node[ann, below left=-2pt] at (0,0) {$\hat{\vec r}_{sat}^{\,j}\,\odot$};
\end{tikzpicture}%
}
\\[3pt]{\footnotesize (b) Aspect angle (plane $\perp$ radial)}
\end{minipage}
\caption{LEO Doppler navigation geometry. \textbf{(a)}~A ground receiver observes two LEO satellites at different altitudes, viewed in the orbital plane. The range rate $\dot{\rho}^j$ is the projection of the relative velocity onto the LOS direction $\hat{\vec\rho}^j$; the LOS rate vector $\dot{\hat{\vec\rho}}^j$ captures the angular motion of the satellite as seen from the receiver. \textbf{(b)}~The satellite velocity in general tilts out of the orbital plane; the aspect angle $\chi^j$ is shown in the plane perpendicular to the satellite radial $\hat{\vec r}_{sat}^{\,j}$ (a change of viewing plane, looking along the radial, $\odot$), measured between the along-track direction $\hat{\vec v}_{sat}^{\,j}$ and the projected LOS $\hat{\vec\rho}_\perp^{\,j}$.}
\label{fig:doppler_geometry}

\end{figure}

\vspace{-10pt}

\subsection{Pseudorange Geometric Dilution of Precision}\label{subsec:pr_gdop}

In pseudorange-based GNSS, the receiver estimates a four-state vector, three position components and receiver clock bias, from $N$ pseudorange measurements through the $N \times 4$ Jacobian matrix $A_{PR} = \begin{bmatrix} A_v & \ones \end{bmatrix}$ \parencite{Kaplan2006,HofmannWellenhof2008}, where $A_v$ is the same matrix of LOS direction vectors as in Eq.~\eqref{eq:jacobian}, entering there as velocity sensitivity and here as position sensitivity, and $\ones$ is the clock bias column. The Geometric Dilution of Precision, defined as $\mathrm{GDOP} = \sqrt{\mathrm{Tr}[(A_{PR}^T A_{PR})^{-1}]}$, can be expressed in terms of the Gram matrix determinant and adjugate as \parencite{BlancoDelgado2017}
\begin{equation}\label{eq:gdop_det}
\mathrm{GDOP} = \sqrt{\frac{\sum_{k=1}^{4}
  [\mathrm{adj}(A_{PR}^T A_{PR})]_{kk}}{\det(A_{PR}^T A_{PR})}}
\end{equation}
where $[\cdot]_{kk}$ denotes the $k$th diagonal entry and the sum runs over the four navigation states. Minimizing GDOP is commonly approximated by maximizing $\det(A_{PR}^T A_{PR})$, because the adjugate trace varies less strongly with geometry than the determinant \parencite{Parkinson1996,BlancoDelgado2017}. For $N = 4$ satellites, the square root of the determinant is proportional to the volume of the tetrahedron whose vertices are the tips of the four LOS unit vectors \parencite{MassattRudnick1990}, establishing the volume-based selection heuristics used for pseudorange GDOP minimization \parencite[see also][for multi-constellation extensions]{BlancoDelgado2010}. The proxy is itself imperfect even in the pseudorange case: \textcite{BlancoDelgado2017} show that the volume--GDOP relationship weakens as the satellite count grows and is tight only near a regular polytope. The point here is not that volume is an exact pseudorange proxy, but that whatever validity it has rests on a clock bias column that is independent of geometry -- a property the Doppler problem does not share.

\vspace{-10pt}

\subsection{The Clock Bias Column}\label{subsec:partition}
The clock bias element for satellite $j$ is
\begin{equation}\label{eq:coupling_column}
a_b^j = \frac{(\hat{\vec\rho}^j)^T \dot{\bar{\vec{v}}}^j + (\dot{\hat{\vec\rho}}^j)^T \bar{\vec{v}}^j}{\eta}
\end{equation}
When analyzed as a function of satellite geometry, this element is denoted $c^j \equiv a_b^j$; a closed-form expression for $c^j$ in terms of elevation, altitude, and the orientation of the LOS relative to the orbital plane is derived in Section~\ref{sec:parameterization}. Notice that, as opposed to pseudorange-based positioning, in this case the clock bias sensitivity depends on satellite velocity, acceleration, and the LOS geometry, which also determine the position and velocity Jacobian blocks. This geometry dependence gives rise to correlations between clock bias and the remaining navigation states, referred to throughout as the clock-bias coupling, in analogy with the clock--height coupling of pseudorange positioning.

The clock-bias element $c^j$ of Eq.~\eqref{eq:coupling_column} has a direct physical origin. The Doppler measurement depends on the satellite position and velocity, evaluated at a time set by the receiver clock; a receiver clock-bias error, therefore, offsets that evaluation time, and because the satellite is moving, the LOS geometry shifts. Consistent with this mechanism, the satellite state is evaluated at the signal transmit time inferred from the receiver's local clock; were the geometry instead evaluated at a true, clock-independent epoch, this sensitivity would vanish. This is why a single-epoch Doppler observation is sensitive to the receiver clock bias even though the measurement model of Eq.~\eqref{eq:simplified_doppler} contains only the clock drift explicitly: the first-order sensitivity is the clock bias column of Eq.~\eqref{eq:coupling_column}, built from the relative acceleration and the LOS rate, the analytical form of which was given by \textcite{Psiaki2021}. The two projections sum to the negative of the normalized range acceleration, $a_b^j = -\ddot{\rho}^j/\eta$, where $\ddot{\rho}^j$ is the second time derivative of the slant range $\rho^j$; this is the reading given by \textcite{Psiaki2021}, who described this column as the negatives of the range accelerations and constructed $\eta$ as an upper bound on their magnitude. This sensitivity scales with how fast the geometry is changing, that is, with $|\dot{\hat{\vec\rho}}^j|$, so satellites with faster LOS rotation are more sensitive to clock-bias errors. The position sensitivity is $\dot{\hat{\vec\rho}}^j/\gamma$, which scales with the same quantity; consequently a receiver position error $\Delta\vec{r}$, entering through $(\dot{\hat{\vec\rho}}^j)^T\Delta\vec{r}$, and a clock-bias error, entering through a term proportional to $|\dot{\hat{\vec\rho}}^j|^2$, are difficult to distinguish from a single Doppler measurement. This shared dependence on the LOS rate is the physical origin of the correlation between the clock bias column and the position columns. The analysis of this structural correlation is one of the subjects of the remainder of this paper.

\vspace{-10pt}

\section{Geometric Parameterization of the Doppler Jacobian}\label{sec:parameterization}

\vspace{-5pt}

This section derives closed-form expressions for all elements of the Doppler Jacobian in terms of satellite elevation, azimuth, inclination, and orbital altitude ratio. The key result is a closed form expression for the clock-bias sensitivity in terms of elevation, altitude, and the velocity--LOS aspect angle, together with the two envelopes that bound its values. A latitude-independence property 
is then established that makes the coupling structure universal across receiver locations.

\vspace{-10pt}

\subsection{Closed-Form Jacobian Elements}\label{subsec:closed_form}

Throughout this section the LOS geometry is resolved in an East--North--Up (ENU) topocentric frame at the receiver location: the $x$-axis points East, the $y$-axis points toward geographic North, and the $z$-axis points to the local zenith. The unit LOS vector toward satellite $j$ is parameterized by its elevation $\varepsilon^j$ (measured from the local horizontal plane, positive toward zenith) and azimuth $\alpha^j$ (measured clockwise from North),
\begin{equation}\label{eq:los_enu}
\hat{\vec\rho}^j =
\begin{bmatrix} \cos\varepsilon^j\sin\alpha^j \\ \cos\varepsilon^j\cos\alpha^j \\ \sin\varepsilon^j \end{bmatrix}_{\mathrm{ENU}},
\qquad z^j \equiv \sin\varepsilon^j = \hat{\vec\rho}_z^j .
\end{equation}
The altitude ratio is $\beta^j \equiv R_E/a_{orb}^j$.

All elements of the Jacobian in Eq.~\eqref{eq:jacobian} can be expressed in closed form as functions of the four geometric parameters ($\varepsilon^j,\alpha^j,i^j,\beta^j$), where $i^j$ is the orbital inclination of satellite $j$, thus extending the measurement model in \textcite{Psiaki2021} to enable direct geometric analysis \parencite[see also][for exact analytical differentiation of the high-fidelity Doppler Jacobian]{Baron2024}. The receiver, the Earth center, and the satellite form a triangle with sides $R_E$ (Earth center $O$ to receiver $R$), $a_{orb}^j$ ($O$ to satellite $S$), and $\rho^j$ ($R$ to $S$, the slant range), shown in Figure~\ref{fig:slant_triangle}. The interior angle of this triangle at the receiver, between the side $RO$ directed toward the Earth center and the LOS $RS$, equals the supplement of the zenith angle, $\tfrac{\pi}{2}+\varepsilon^j$ -- the zenith angle itself, between the upward vertical and the LOS, being $\tfrac{\pi}{2}-\varepsilon^j$ -- so that $\cos\!\left(\tfrac{\pi}{2}+\varepsilon^j\right) = -\sin\varepsilon^j = -z^j$. The Law of Cosines applied at the receiver gives
\vspace{-10pt}
\begin{figure}[htb]
\centering
\scalebox{0.6}{%
\begin{tikzpicture}[
    >=Stealth,
    sat/.style={fill=gray!70, draw=black, line width=0.4pt,
                minimum size=8pt, inner sep=0pt, rectangle},
    ann/.style={font=\LARGE},
  ]
  \def\RE{3.0}                          
  \def\hSat{4.2}                        
  \def\elev{28}                         
  \pgfmathsetmacro{\posS}{90-\elev-20}  

  \coordinate (O) at (0,0);             
  \coordinate (R) at (90:\RE);          
  \coordinate (S) at (\posS:{\RE+\hSat});

  \coordinate (Heast) at ($(R)+(2.4,0)$);
  \coordinate (Hwest) at ($(R)+(-1.3,0)$);

  \fill[blue!5] (42:\RE) arc[start angle=42, end angle=138, radius=\RE] -- (O) -- cycle;
  \draw[thick, blue!30!black] (42:\RE) arc[start angle=42, end angle=138, radius=\RE];

  \draw[dashed, gray!55] (Hwest) -- (Heast)
     node[pos=0.0, left=1pt, ann, gray!45!black] {horizon};
  \coordinate (Up) at (90:{\RE+2.0});
  \draw[->, dashed, gray!55!black] (R) -- (Up) node[right=2pt, ann] {local zenith (Up)};

  \draw[thick] (O) -- (R) node[midway, left, ann] {$R_E$};
  \draw[thick] (O) -- (S) node[pos=0.55, below right=-1pt, ann] {$a_{\mathrm{orb}}^j$};
  \draw[thick] (R) -- (S) node[pos=0.5, above, ann] {$\rho^j$};

  \pgfmathanglebetweenpoints{\pgfpointanchor{R}{center}}{\pgfpointanchor{S}{center}}
  \let\angLOS\pgfmathresult
  \draw[red!70!black, line width=0.7pt]
     ([shift=(0:1.25cm)]R) arc[start angle=0, end angle=\angLOS, radius=1.25cm];
  \pgfmathsetmacro{\elevmid}{\angLOS/2}
  \node[red!70!black, ann] at ([shift=(\elevmid:1.8cm)]R) {$\varepsilon^j$};
  \draw[black, line width=0.6pt]
     ([shift=(\angLOS:0.62cm)]R) arc[start angle=\angLOS, end angle=90, radius=0.62cm];
  \pgfmathsetmacro{\zenmid}{(\angLOS+90)/2}
  \node[ann, font=\LARGE] at ([shift=(\zenmid:1.45cm)]R) {$\tfrac{\pi}{2}-\varepsilon^j$};

  \fill[black] (O) circle (1.6pt);
  \node[ann, below=2pt] at (O) {Earth center $O$};
  \fill[black] (R) circle (2.2pt);
  \node[ann, above left=0pt, font=\LARGE\bfseries] at (R) {$R$};
  \node[sat] at (S) {};
  \node[ann, right=2pt] at (S) {Satellite $S$};

\end{tikzpicture}%
}\caption{Earth-center/receiver/satellite triangle defining the normalized slant range.}
\label{fig:slant_triangle}
\vspace{-20pt}
\end{figure}
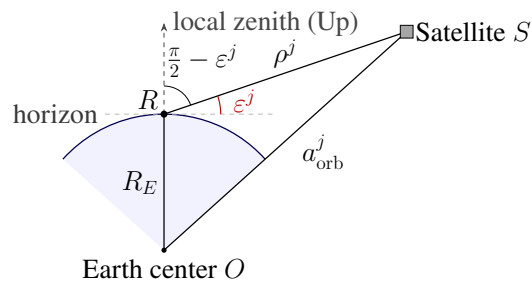
\begin{equation}\label{eq:law_of_cosines}
(a_{orb}^j)^2 = R_E^2 + (\rho^j)^2 - 2 R_E\,\rho^j\cos\!\left(\tfrac{\pi}{2}+\varepsilon^j\right)
             = R_E^2 + (\rho^j)^2 + 2 R_E\,\rho^j z^j .
\end{equation}
Dividing by $(a_{orb}^j)^2$ and introducing the normalized slant range $\kappa^j \equiv \rho^j/a_{orb}^j$ yields the quadratic $(\kappa^j)^2 + 2\beta^j z^j\,\kappa^j + \big((\beta^j)^2 - 1\big) = 0$, and therefore the normalized slant range is
\begin{equation}\label{eq:kappa}
\kappa(z;\beta) = \sqrt{\beta^2 z^2 + (1 - \beta^2)} - \beta z .
\end{equation}
Dropping the explicit dependence of $\kappa$ on the local elevation $z$ and the altitude ratio $\beta$, re-write Eq.~\eqref{eq:kappa} as
\begin{equation}\label{eq:kappa_identity}
1 = \beta^2 + (\kappa)^2 + 2\beta z\,\kappa ,
\end{equation}
repeatedly used in what follows. Notice that for a satellite at zenith ($z=1$), $\kappa = 1-\beta$ and at the horizon ($z=0$), $\kappa = \sqrt{1-\beta^2}$.

Each row $a_v^j$ of the velocity block depends on $(\varepsilon^j, \alpha^j)$ alone through the corresponding LOS direction $\hat{\vec\rho}^j$ (Eq.~\eqref{eq:los_enu}), identical to the rows of the pseudorange geometry matrix. Because $\hat{\vec\rho}^j$ depends on the receiver-to-satellite direction and not on the satellite's orbit, the velocity block carries no information about inclination, pass direction, or altitude.

The rows $a_p^j$ of the position block involve the LOS rate vector $\dot{\hat{\vec\rho}}^j$, the time derivative of the unit LOS direction. Differentiating $\hat{\vec\rho}^j = \vec\rho^j/\rho^j$ with the quotient rule, and using $\dot\rho^j = (\hat{\vec\rho}^j)^T\dot{\vec\rho}^j$ for the range rate, the radial part of the motion cancels and only the component of the relative velocity transverse to the LOS survives:
\begin{equation}\label{eq:rho_dot_fundamental}
    \dot{\hat{\vec\rho}}^j
    = \frac{\dot{\vec\rho}^j}{\rho^j} - \frac{\dot\rho^j}{\rho^j}\hat{\vec\rho}^j
    = \frac{1}{\rho^j}\Big(\dot{\vec\rho}^j - \big((\hat{\vec\rho}^j)^T \dot{\vec\rho}^j\big)\hat{\vec\rho}^j\Big)
    = \frac{\vec{v}_{rel,\perp}^j}{\rho^j} .
\end{equation}
Here, $\vec\rho^j$ is the receiver-to-satellite relative position vector, $\rho^j = \|\vec\rho^j\|$ is the slant range (so that $\hat{\vec\rho}^j = \vec\rho^j/\rho^j$), $\dot{\vec\rho}^j$ is the relative velocity between satellite $j$ and the receiver, and $\vec{v}_{rel,\perp}^j$ denotes its component perpendicular to the LOS (cf. Fig.~\ref{fig:doppler_geometry}).  Geometrically, the LOS direction can only rotate; its rate of rotation is set by how fast the satellite moves across the sky, i.e. by the transverse velocity, while motion directly along the LOS changes the range but not the pointing direction. From the kinematic relation $\dot{\hat{\vec\rho}}^j = \vec v_{rel,\perp}^j/\rho^j$ (Eq.~\eqref{eq:rho_dot_fundamental}), the magnitude is $\|\dot{\hat{\vec\rho}}^j\| = v_{sat}^j\sin\alpha_v^j/\rho^j$, where $v_{sat}^j$ is the satellite orbital speed and $\alpha_v^j$ the angle between the satellite velocity and the LOS. With the circular-orbit relations $v_{sat}^j = \sqrt{\mu/a_{orb}^j}$, $\rho^j = a_{orb}^j\kappa^j$, and $n_{orb}^j = \sqrt{\mu/(a_{orb}^j)^3}$, this becomes $\|\dot{\hat{\vec\rho}}^j\| = n_{orb}^j\sin\alpha_v^j/\kappa^j$, with $\kappa^j$ given by Eq.~\eqref{eq:kappa}. The angle $\alpha_v^j$ is not fixed by elevation and altitude alone; it also depends on the orientation of the LOS relative to the orbital plane. This orientation is measured by the aspect angle $\chi^j\in[0,\pi/2]$, defined as the angle, in the plane perpendicular to the Earth-center-to-satellite radial, between the satellite velocity and the projection of the LOS onto that plane (Fig.~\ref{fig:doppler_geometry}(b)). The aspect angle is zero for a LOS in the orbital plane and $\pi/2$ for a LOS cross-track to the velocity. The LOS-rate magnitude, therefore, depends on $(z^j,\beta^j,\chi^j)$. Writing $\gamma^j = n_{orb}^j/(1-\beta^j)$ for the scaling parameter of Eq.~\eqref{eq:gamma_eta}, we define the \emph{normalized LOS-rate magnitude}
\begin{equation}\label{eq:S_scaling}
\tilde{S}(z;\beta,\chi) \;\equiv\; \frac{\|\dot{\hat{\vec\rho}}\|}{\gamma} \;=\; \frac{(1-\beta)\sin\alpha_v(z;\beta,\chi)}{\kappa(z;\beta)} .
\end{equation}
%
At zenith, the transverse velocity carries the full orbital speed independently of the heading, so $\tilde S=1$; for lower elevations, $\tilde S$ decreases and becomes heading-dependent, ranging within a band whose lower edge is the in-plane value $(1-\beta)(\beta z+\kappa)/\kappa$ and whose width is set by the aspect angle $\chi$. Per-measurement position information is therefore reduced for low-elevation satellites.

The direction of $\dot{\hat{\vec\rho}}^j$ within the plane perpendicular to $\hat{\vec\rho}^j$ gives the satellite's apparent motion across the sky. 
From Eq.~\eqref{eq:rho_dot_fundamental}, this direction is perpendicular to the LOS,
so the apparent heading is obtained by resolving $\dot{\hat{\vec\rho}}^j$ in a basis of the sky plane. The plane perpendicular to $\hat{\vec\rho}^j$ is spanned by an orthonormal pair: the in-sky cross-track unit vector $\vec{e}^j = (\hat{\vec n}_0\times\hat{\vec\rho}^j)/\|\hat{\vec n}_0\times\hat{\vec\rho}^j\|$, built from the local North direction $\hat{\vec n}_0 = [0,1,0]^T_{\mathrm{ENU}}$, and the in-sky along-track unit vector $\vec{n}^j = \hat{\vec\rho}^j\times\vec{e}^j$; the pair is well-defined except on the due-north or due-south horizon, where $\hat{\vec\rho}^j$ is parallel to $\hat{\vec n}_0$ and the heading is assigned by continuity. The apparent heading $\psi^j$, measured clockwise from North, is then defined directly as the angle of the projected relative velocity in this basis,
\begin{equation}\label{eq:psi_app}
\psi^j \equiv \operatorname{atan2}\!\big((\dot{\hat{\vec\rho}}^j)^T\vec{e}^j,\ (\dot{\hat{\vec\rho}}^j)^T\vec{n}^j\big),
\end{equation}
so that $\dot{\hat{\vec\rho}}^j/\|\dot{\hat{\vec\rho}}^j\| = \cos\psi^j\,\vec{n}^j + \sin\psi^j\,\vec{e}^j$. For a circular orbit Clairaut's relation \parencite{Geyer2016} supplies the ground-track heading $\psi_{gt}^j$, the direction of the satellite velocity at the sub-satellite point measured clockwise from North, $\sin\psi_{gt}^j = \cos i^j/\cos\phi_{sat}^j$, with $\cos\psi_{gt}^j = \sigma^j\sqrt{1-\sin^2\psi_{gt}^j}$ and $\sigma^j\in\{-1,+1\}$ the pass-direction indicator ($+1$ ascending, $-1$ descending). The satellite latitude $\phi_{sat}^j$ is obtained by projecting the satellite position onto the Earth rotation axis $\hat{\vec K}_{E} = [0,0,1]^T_{\mathrm{ECI}}$, an axis that coincides with the local North direction only for an equatorial receiver. The heading $\psi_{gt}^j$ fixes the satellite velocity vector, from which Eq.~\eqref{eq:psi_app} delivers $\psi^j$. The complete row of the eight-state Jacobian for satellite $j$ can then be written as:
\begin{equation}\label{eq:8state_row}
\begin{bmatrix}
\tilde{S}^j\left[\cos\psi^j\,\vec{n}^j + \sin\psi^j\,\vec{e}^j\right]^T &
c^j &
(\hat{\vec\rho}^j)^T &
1
\end{bmatrix}
\end{equation}
The position block is the product of a scalar magnitude $\tilde{S}^j$, the normalized LOS-rate magnitude of Eq.~\eqref{eq:S_scaling} that measures how fast the LOS sweeps across the sky, and a unit vector along the apparent motion across the sky. The magnitude depends on elevation, altitude, and, through the aspect angle $\chi$, the satellite heading; the direction depends on inclination and pass direction. The clock-drift entry is unity. The clock-bias entry $c^j$ is the critical structural element and is derived next; it is built from the same LOS rate that sets the position block and, therefore, shares the position block's dependence on the LOS-rate magnitude. This shared dependence is the origin of the structural collinearity analyzed in Sections~\ref{sec:coupling} and~\ref{sec:schur}.

The clock-bias element $c^j = a_b^j$ of Eq.~\eqref{eq:coupling_column} is the sum of two projections: the relative acceleration onto the LOS, and the relative velocity onto the LOS rate. We evaluate each in closed form. For a stationary receiver the relative velocity and acceleration are $\bar{\vec v}^j = -\vec v^j$ and $\dot{\bar{\vec v}}^j = -\dot{\vec v}^j$, respectively, and the satellite acceleration is centripetal, $\dot{\vec v}^j = -(\mu/(a_{orb}^j)^2)\,\hat{\vec r}_{sat}^j$, where $\hat{\vec r}_{sat}^j$ is the unit vector from the Earth center to the satellite.

\vspace{-10pt}

\subsubsection{First projection -- relative acceleration onto the LOS}
In ENU coordinates, the satellite position can be written as $\vec r_{sat}^j = R_E\,\hat{\vec z}_{\mathrm{ENU}} + \rho^j\hat{\vec\rho}^j$ so that:
\begin{equation}\label{eq:rho_dot_rsat}
\hat{\vec\rho}^j\cdot\vec r_{sat}^j = R_E z^j + \rho^j = a_{orb}^j\,(\beta^j z^j + \kappa^j)\Rightarrow
\hat{\vec\rho}^j\cdot\hat{\vec r}_{sat}^j = \beta^j z^j + \kappa^j .
\end{equation}
As a check, at zenith the satellite lies directly overhead, so the LOS and the Earth-center-to-satellite radial direction coincide; setting $z^j=1$ and $\kappa^j=1-\beta^j$ indeed gives $\hat{\vec\rho}^j\cdot\hat{\vec r}_{sat}^j = \beta^j + (1-\beta^j) = 1$. Substituting the centripetal acceleration $\dot{\bar{\vec v}}^j = -\dot{\vec v}^j = (\mu/(a_{orb}^j)^2)\,\hat{\vec r}_{sat}^j$ and normalizing by $\eta^j = (\beta^j/(1-\beta^j))\,\mu/(a_{orb}^j)^2$,
\begin{equation}\label{eq:term1}
\underbrace{\frac{(\hat{\vec\rho}^j)^T\dot{\bar{\vec v}}^j}{\eta^j}}_{\text{Term 1}}
= \frac{(\mu/(a_{orb}^j)^2)\,(\beta^j z^j + \kappa^j)}{(\beta^j/(1-\beta^j))\,\mu/(a_{orb}^j)^2}
= \frac{1-\beta^j}{\beta^j}\,(\beta^j z^j + \kappa^j).
\end{equation}
\vspace{-10pt}
\subsubsection{Second projection -- relative velocity onto the LOS rate} Using $\dot{\hat{\vec\rho}}^j = \vec v_{sat,\perp}^j/\rho^j$ (Eq.~\eqref{eq:rho_dot_fundamental}), where $\vec v_{sat,\perp}^j$ is the component of satellite velocity perpendicular to the LOS,
\begin{equation}\label{eq:term2_exact_proj}
(\dot{\hat{\vec\rho}}^j)^T\bar{\vec v}^j
= -(\dot{\hat{\vec\rho}}^j)^T\vec v^j
= -\frac{\|\vec v_{sat,\perp}^j\|^2}{\rho^j}
= -\frac{(v_{sat}^j)^2\sin^2\!\alpha_v^j}{\rho^j},
\end{equation}
where $\alpha_v^j$ is the angle between the satellite velocity and the LOS, so that $\|\vec v_{sat,\perp}^j\| = v_{sat}^j\sin\alpha_v^j$. The factor $\sin^2\!\alpha_v^j$ has a closed form once the orientation of the LOS relative to the orbital plane is resolved. Expand $\hat{\vec\rho}^j$ in the orbit-aligned orthonormal triad at the satellite -- the radial direction $\hat{\vec r}_{sat}^j$, the along-track direction $\hat{\vec v}_{sat}^j$, and the orbit normal $\hat{\vec h}^j = \hat{\vec r}_{sat}^j\times\hat{\vec v}_{sat}^j$. The radial coefficient is $\hat{\vec\rho}^j\cdot\hat{\vec r}_{sat}^j = \beta^j z^j+\kappa^j$ from Eq.~\eqref{eq:rho_dot_rsat}, and the identity~\eqref{eq:kappa_identity} gives $1-(\beta^j z^j+\kappa^j)^2 = (\beta^j)^2\cos^2\varepsilon^j$. Because $\hat{\vec\rho}^j$ has unit norm, the squared coefficients on the triad sum to one, so that $\cos^2\!\alpha_v^j = (\beta^j)^2\cos^2\varepsilon^j - (\hat{\vec\rho}^j\cdot\hat{\vec h}^j)^2$. The orbit-normal coefficient $\hat{\vec\rho}^j\cdot\hat{\vec h}^j$ is the out-of-orbital-plane component of the LOS; with the aspect angle $\chi^j$ defined above and $\|\hat{\vec\rho}_\perp^j\| = \beta^j\cos\varepsilon^j$, $|\hat{\vec\rho}^j\cdot\hat{\vec h}^j| = \beta^j\cos\varepsilon^j\sin\chi^j$ and $\cos\alpha_v^j = \beta^j\cos\varepsilon^j\cos\chi^j$. Equivalently, the aspect angle has the explicit form $\chi^j = \operatorname{atan2}\!\big(|\hat{\vec\rho}^j\cdot\hat{\vec h}^j|,\ |\hat{\vec\rho}^j\cdot\hat{\vec v}_{sat}^j|\big) \in [0,\pi/2]$, the out-of-plane tilt of the LOS read directly from its orbit-normal and along-track projections. Like the apparent heading $\psi^j$, the aspect angle is fixed by the orbital orientation through the ground-track heading $\psi_{gt}^j$. This gives
\begin{equation}\label{eq:sin2_alpha}
\sin^2\!\alpha_v^j = 1 - (\beta^j)^2\cos^2\varepsilon^j\cos^2\chi^j
= (\beta^j z^j + \kappa^j)^2 + (\beta^j)^2\cos^2\varepsilon^j\sin^2\chi^j .
\end{equation}
The second equality uses the identity~\eqref{eq:kappa_identity}, $(\beta^j z^j+\kappa^j)^2 = 1-(\beta^j)^2\cos^2\varepsilon^j$. With $(v_{sat}^j)^2 = \mu/a_{orb}^j$ and $\rho^j = a_{orb}^j\kappa^j$, Eq.~\eqref{eq:term2_exact_proj} becomes
\begin{equation}\label{eq:term2}
\underbrace{\frac{(\dot{\hat{\vec\rho}}^j)^T\bar{\vec v}^j}{\eta^j}}_{\text{Term 2}}
= -\frac{1-\beta^j}{\beta^j}\,\frac{\sin^2\!\alpha_v^j}{\kappa^j} .
\end{equation}
\subsubsection{Clock-bias sensitivity}
Adding Eqs.~\eqref{eq:term1} and~\eqref{eq:term2} and simplifying with identity~\eqref{eq:kappa_identity},
\begin{equation}\label{eq:coupling_analytical}
c(z;\beta,\chi) = -\,\frac{1-\beta}{\kappa}\Big[\,z\,(\beta z + \kappa) + \beta\cos^2\varepsilon\,\sin^2\chi\,\Big],
\qquad \cos^2\varepsilon = 1-z^2 .
\end{equation}
For a satellite above the horizon, $z$ is nonnegative and the normalized slant range $\kappa$ is positive. Because $0 < \beta < 1$, it follows that $c(z;\beta,\chi) \le 0$
for all elevation and aspect angles. The closed form is exact within the simplified GDOP model in Section~\ref{subsec:doppler_model}.
Relative to a full operational Doppler Jacobian evaluated on realistic ephemerides for an Earth-fixed receiver, neglecting the Earth-rotation velocity of the receiver (about 0.47~km/s, roughly six percent of the orbital speed) and the orbital eccentricity is a geometric simplification at the few-percent level, retained because it leaves the column structure under study intact. The sensitivity is linear in $\sin^2\chi$ with a nonpositive slope, so that at fixed $(z,\beta)$ its magnitude is bracketed by two envelopes, the in-plane floor ($\chi=0$) and the cross-track ceiling ($\chi=\pi/2$),
\begin{equation}\label{eq:coupling_envelopes}
c_{\mathrm{ip}}(z;\beta) = -\,\frac{(1-\beta)\,z\,(\beta z+\kappa)}{\kappa},
\qquad
c_{\mathrm{xt}}(z;\beta) = -\,\frac{(1-\beta)\,(z\kappa+\beta)}{\kappa},
\end{equation}
so that $|c_{\mathrm{ip}}(z;\beta)| \le |c(z;\beta,\chi)| \le |c_{\mathrm{xt}}(z;\beta)|$. At $\chi=0$, the LOS lies in the orbital plane, $\sin^2\!\alpha_v = (\beta z+\kappa)^2$, and the sensitivity reduces to the compact in-plane form
\begin{equation}\label{eq:c_eq_zStilde}
c_{\mathrm{ip}}(z;\beta) = -\,z\,\tilde S(z;\beta,0) , \qquad \tilde S(z;\beta,0) = \frac{(1-\beta)\,(\beta z+\kappa)}{\kappa} ,
\end{equation}
the negative elevation sine times the in-plane LOS-rate magnitude. More generally, the aspect-angle-dependent part of the sensitivity, $-\,(1-\beta)\beta\cos^2\varepsilon\,\sin^2\chi/\kappa$, is the projection of the relative velocity onto the LOS rate (Term~2, proportional to $\sin^2\!\alpha_v$) and, therefore, lies in the span of the position columns, whose magnitude is the same $\tilde S^j$.

At zenith 
the aspect-angle term vanishes and $c=-1$ at every altitude and aspect angle. At the horizon ($z=0$) only the aspect-angle term survives, $c=-(1-\beta)\beta\sin^2\chi/\kappa$, which vanishes only for a LOS that also lies in the orbital plane, so a cross-track horizon-grazing satellite retains a nonzero sensitivity. The per-satellite bound $|c(z;\beta,\chi)|\le 1$ holds at every aspect angle, with equality only at zenith, consistent with the definition of $\eta$ as the zenith range acceleration. The bound applies in each satellite's own normalization: when satellites at different altitudes are combined under a single common normalization, the rescaled entries can exceed unity, and the aggregate $d=\|\mathbf{c}\|^2$ can exceed the satellite count. The anatomy of the two projections that compose $c^j$, and which navigation states they couple to, is taken up in Section~\ref{subsec:coupling_structure}.

The term $c^j$ includes the elevation factor $z^j$, the $z$-component of the velocity columns in the Jacobian. At fixed altitude and heading, $c^j$ is a function of $z^j$; across a constellation the per-satellite sensitivity is confined to the band $[\,c_{\mathrm{ip}}(z;\beta),\,c_{\mathrm{xt}}(z;\beta)\,]$ between the two envelopes, and its heading-dependent part lies in the span of the position columns. Both channels -- the shared elevation factor with the velocity $z$-column and the shared LOS-rate magnitude with the position columns -- align the clock bias column with the seven-state subspace; their effect on the GDOP is developed in Sections~\ref{sec:coupling} and~\ref{sec:schur}. Figure~\ref{fig:jacobian_blocks} summarizes the dependence structure of the four Jacobian column groups, showing the clock bias column as a function of elevation, altitude, and aspect angle, and its coupling to the velocity $z$-component column.

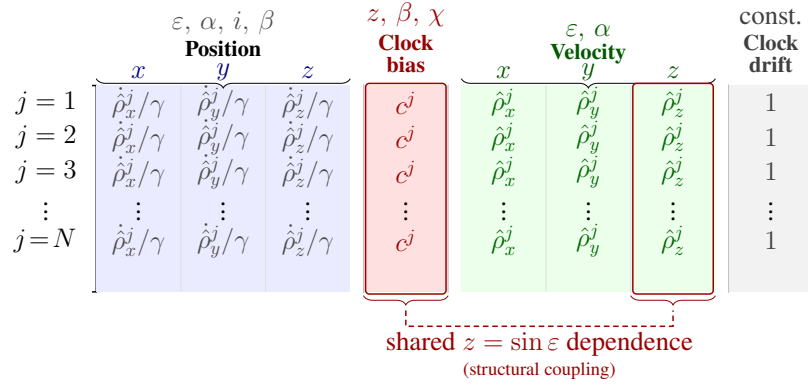
\begin{figure}[htb]
\centering
\scalebox{0.7}{%
\begin{tikzpicture}[
    >=Stealth,
    ann/.style={font=\Large}, 
    colhead/.style={font=\large\bfseries, align=center}, 
    dep/.style={font=\Large, align=center, text=gray!60!black}, 
    brace/.style={decorate, decoration={brace, amplitude=5pt, raise=2pt}},
    bracebelow/.style={decorate, decoration={brace, amplitude=5pt, raise=2pt, mirror}},
  ]
 
  \def\cw{1.60}       
  \def\rh{0.65}      
  \def\nrows{6}      
  \def\gap{0.25}     
 
  \pgfmathsetmacro{\posX}{0}
  \pgfmathsetmacro{\posW}{3*\cw}
  \pgfmathsetmacro{\cbX}{\posW + \gap}
  \pgfmathsetmacro{\cbW}{\cw}
  \pgfmathsetmacro{\velX}{\cbX + \cbW + \gap}
  \pgfmathsetmacro{\velW}{3*\cw}
  \pgfmathsetmacro{\cdX}{\velX + \velW + \gap}
  \pgfmathsetmacro{\cdW}{\cw}
  \pgfmathsetmacro{\totW}{\cdX + \cdW}
  \pgfmathsetmacro{\totH}{\nrows*\rh}
 
  \draw[thick] (-.08, 0) -- (0, 0) -- (0, -\totH) -- (-.08, -\totH);
  \draw[thick] ({\totW+.08}, 0) -- (\totW, 0) -- (\totW, -\totH)
    -- ({\totW+.08}, -\totH);
 
  \fill[blue!8] (\posX, 0) rectangle ({\posX+\posW}, -\totH);
  \fill[red!12] (\cbX, 0) rectangle ({\cbX+\cbW}, -\totH);
  \fill[green!8] (\velX, 0) rectangle ({\velX+\velW}, -\totH);
  \fill[gray!10] (\cdX, 0) rectangle ({\cdX+\cdW}, -\totH);
 
  \foreach \i in {1,2} {
    \draw[gray!25, very thin]
      ({\posX+\i*\cw}, 0) -- ({\posX+\i*\cw}, -\totH);
  }
  \foreach \i in {1,2} {
    \draw[gray!25, very thin]
      ({\velX+\i*\cw}, 0) -- ({\velX+\i*\cw}, -\totH);
  }
 
  \draw[gray!50, thin]
    (\cbX, 0.02) -- (\cbX, {-\totH-0.02});
  \draw[gray!50, thin]
    ({\cbX+\cbW}, 0.02) -- ({\cbX+\cbW}, {-\totH-0.02});
  \draw[gray!50, thin]
    (\cdX, 0.02) -- (\cdX, {-\totH-0.02});
 
  \foreach \r [count=\ri from 0] in {1,2,3} {
    \pgfmathsetmacro{\ry}{-\ri*\rh - \rh/2}
    \node[ann, anchor=east] at (-0.25, \ry) {$j=\r$};
  }
  \pgfmathsetmacro{\vdotY}{-3*\rh - \rh/2}
  \node[ann] at (-0.95, \vdotY) {$\vdots$};
  \pgfmathsetmacro{\nrowY}{-4*\rh - \rh/2}
  \node[ann, anchor=east] at (-0.25, \nrowY) {$j\!=\!N$};
 
  \foreach \ri in {0,1,2,4} {
    \pgfmathsetmacro{\ry}{-\ri*\rh - \rh/2}
    \node[ann, gray!60!black] at ({\posX+0.5*\cw}, \ry)
      {\Large $\dot{\hat{\rho}}_x^j/\gamma$};
    \node[ann, gray!60!black] at ({\posX+1.5*\cw}, \ry)
      {\Large $\dot{\hat{\rho}}_y^j/\gamma$};
    \node[ann, gray!60!black] at ({\posX+2.5*\cw}, \ry)
      {\Large $\dot{\hat{\rho}}_z^j/\gamma$};
  }
  \node[ann] at ({\posX+0.5*\cw}, \vdotY) {$\vdots$};
  \node[ann] at ({\posX+1.5*\cw}, \vdotY) {$\vdots$};
  \node[ann] at ({\posX+2.5*\cw}, \vdotY) {$\vdots$};
 
  \foreach \ri in {0,1,2,4} {
    \pgfmathsetmacro{\ry}{-\ri*\rh - \rh/2}
    \node[ann, red!60!black] at ({\cbX+0.5*\cw}, \ry) {\Large $c^j$};
  }
  \node[ann] at ({\cbX+0.5*\cw}, \vdotY) {$\vdots$};
 
  \foreach \ri in {0,1,2,4} {
    \pgfmathsetmacro{\ry}{-\ri*\rh - \rh/2}
    \node[ann, green!40!black] at ({\velX+0.5*\cw}, \ry)
      {\Large $\hat{\rho}_x^j$};
    \node[ann, green!40!black] at ({\velX+1.5*\cw}, \ry)
      {\Large $\hat{\rho}_y^j$};
    \node[ann, green!40!black] at ({\velX+2.5*\cw}, \ry)
      {\Large $\hat{\rho}_z^j$};
  }
  \node[ann] at ({\velX+0.5*\cw}, \vdotY) {$\vdots$};
  \node[ann] at ({\velX+1.5*\cw}, \vdotY) {$\vdots$};
  \node[ann] at ({\velX+2.5*\cw}, \vdotY) {$\vdots$};
 
  \foreach \ri in {0,1,2,4} {
    \pgfmathsetmacro{\ry}{-\ri*\rh - \rh/2}
    \node[ann, gray!50!black] at ({\cdX+0.5*\cw}, \ry) {\Large $1$};
  }
  \node[ann] at ({\cdX+0.5*\cw}, \vdotY) {$\vdots$};
 
 
  \pgfmathsetmacro{\headY}{0.45}
 
  \draw[brace] ({\posX+\posW}, \headY-0.25) -- (\posX, \headY-0.25);
  \node[colhead] at ({\posX+\posW/2}, \headY+0.25)
    {Position};
  \node[dep] at ({\posX+\posW/2}, \headY+0.75)
    {$\varepsilon,\,\alpha,\,i,\,\beta$};
 
  \node[colhead, red!60!black] at ({\cbX+\cbW/2}, \headY+0.15)
    {Clock\\[-2pt]bias};
  \node[dep, red!60!black] at ({\cbX+\cbW/2}, \headY+0.9)
    {$z,\,\beta,\,\chi$};
 
  \draw[brace] ({\velX+\velW}, \headY-0.25) -- (\velX, \headY-0.25);
  \node[colhead, green!35!black] at ({\velX+\velW/2}, \headY+0.15)
    {Velocity};
  \node[dep, green!35!black] at ({\velX+\velW/2}, \headY+0.55)
    {$\varepsilon,\,\alpha$};
 
  \node[colhead, gray!50!black] at ({\cdX+\cdW/2}, \headY+0.15)
    {Clock\\[-2pt]drift};
  \node[dep, gray!50!black] at ({\cdX+\cdW/2}, \headY+0.85)
    {const.};
 
  \pgfmathsetmacro{\subY}{0.25}
  \node[font=\Large, green!35!black] at ({\velX+0.5*\cw}, \subY) {$x$};
  \node[font=\Large, green!35!black] at ({\velX+1.5*\cw}, \subY) {$y$};
  \node[font=\Large, green!35!black] at ({\velX+2.5*\cw}, \subY) {$z$};
 
  \node[font=\Large, blue!50!black] at ({\posX+0.5*\cw}, \subY) {$x$};
  \node[font=\Large, blue!50!black] at ({\posX+1.5*\cw}, \subY) {$y$};
  \node[font=\Large, blue!50!black] at ({\posX+2.5*\cw}, \subY) {$z$};
 
  \pgfmathsetmacro{\hlpad}{0.04}
  \draw[red!60!black, thick, rounded corners=2pt]
    ({\velX+2*\cw+\hlpad}, \hlpad)
    rectangle ({\velX+3*\cw-\hlpad}, {-\totH-\hlpad});
  \draw[red!60!black, thick, rounded corners=2pt]
    ({\cbX+\hlpad}, \hlpad)
    rectangle ({\cbX+\cbW-\hlpad}, {-\totH-\hlpad});
 
  \pgfmathsetmacro{\annY}{-\totH - 0.35}
  \pgfmathsetmacro{\annYtext}{-\totH - 0.9}
 
  \coordinate (cbBot) at ({\cbX+\cbW/2}, \annY);
  \coordinate (vzBot) at ({\velX+2.5*\cw}, \annY);
 
  \draw[bracebelow, red!60!black]
    ({\cbX}, {-\totH}) -- ({\cbX+\cbW}, {-\totH});
  \draw[bracebelow, red!60!black]
    ({\velX+2*\cw}, {-\totH}) -- ({\velX+3*\cw}, {-\totH});
 
  \pgfmathsetmacro{\ubot}{-\totH - 0.65}
  \draw[red!60!black, thick, densely dashed]
    (cbBot) -- ({\cbX+\cbW/2}, \ubot)
    -- ({\velX+2.5*\cw}, \ubot) -- (vzBot);
 
  \node[font=\Large, red!60!black, align=center]
    at ({(\cbX+\cbW/2 + \velX+2.5*\cw)/2}, {-\totH - 1.2})
    {shared $z = \sin\varepsilon$ dependence\\[-5pt]
     \normalsize(structural coupling)};
 
 
\end{tikzpicture}%
}\caption{Block structure of the eight-state Doppler Jacobian. The position columns depend on elevation, azimuth, inclination, and altitude; the velocity columns depend on elevation and azimuth; the clock bias column depends on elevation, altitude, and the velocity--LOS aspect angle; the clock drift column is constant. The shared elevation dependence of the clock bias and velocity $z$-components is the structural origin of the coupling analyzed in this paper.}
\label{fig:jacobian_blocks}
\vspace{-10pt}
\end{figure}

\vspace{-10pt}

\subsection{Geometric Invariances}\label{subsec:invariances}

The realized sensitivity $c^j$ depends on the aspect angle $\chi^j$ and so varies with the orbital orientation, but the band within which it must lie does not. The two envelopes that bound the sensitivity are functions of elevation and altitude alone, so the same band applies at every receiver location; only the value within it is constrained by the receiver latitude, through the admissible aspect angles. This separation is stated next.

\begin{proposition}[Latitude Independence of the Sensitivity Band]\label{prop:latitude_independence}
For every satellite above the horizon, the clock-bias sensitivity $c(z;\beta,\chi)$ of Eq.~\eqref{eq:coupling_analytical} satisfies $|c_{\mathrm{ip}}(z;\beta)| \le |c(z;\beta,\chi)| \le |c_{\mathrm{xt}}(z;\beta)|$, where the envelopes $c_{\mathrm{ip}}$ and $c_{\mathrm{xt}}$ of Eq.~\eqref{eq:coupling_envelopes} depend only on the local elevation $z = \sin\varepsilon$ and the altitude ratio $\beta = R_E/a_{orb}$. The band is therefore independent of the receiver latitude $\phi_{rec}$, the satellite azimuth $\alpha$, and the orbital inclination $i$.
\end{proposition}

\begin{proof}
The only quantities appearing in the envelopes of Eq.~\eqref{eq:coupling_envelopes} are $z$, $\beta$, and the normalized slant range $\kappa$, which by Eq.~\eqref{eq:kappa} is a function of $(z,\beta)$ alone, because the Earth-center/receiver/satellite triangle has the local elevation as its only geometric degree of freedom. Both envelopes are thus determined by $(z,\beta)$ alone, so the band they bound is independent of receiver latitude $\phi_{rec}$, satellite azimuth $\alpha$, and orbital inclination $i$. The bracketing $|c_{\mathrm{ip}}| \le |c| \le |c_{\mathrm{xt}}|$ holds because $c$ is linear in $\sin^2\chi$ with a nonpositive slope, as established in Section~\ref{subsec:closed_form}.
\end{proof}

Receiver latitude does not enter the envelopes, but it does constrain the aspect angle realized within the band. The inclinations that reach satellite latitude $\phi_{sat}^j$ satisfy $\min(i^j, \pi - i^j) \ge |\phi_{sat}^j|$ \parencite{Montenbruck2000}, which the receiver latitude limits through the sky positions it can produce; a near-polar shell viewed near the equator moves almost meridionally, $\chi^j \approx 0$, placing its satellites near the in-plane floor, whereas an inclined shell admits a wide range of $\chi^j$. What is universal across receiver locations is therefore the sensitivity band, not the value within it. The alignment of the clock bias column with the seven-state subspace is intrinsic to the Doppler measurement physics, not an artifact of receiver location, although its magnitude varies with the admissible aspect angles.

\vspace{-10pt}

\section{Clock-Bias Coupling and the Observability--Correlation Tradeoff}\label{sec:coupling}
This section establishes the structural correlation between the clock bias column and the remaining navigation states, and shows that part of it cannot be removed by any choice of geometry. We first separate the navigation problem into a seven-state formulation and the full eight-state formulation that adds clock bias; this separation is what allows the cost of estimating clock bias to be isolated, through the Schur complement developed in Section~\ref{sec:schur}. We then analyze the two terms that compose the clock-bias sensitivity, showing that it is negative for every satellite above the horizon and is, therefore, irreducibly correlated with the clock drift column. The section closes by characterizing which couplings are intrinsic to the measurement physics and which depend on the specific constellation geometry, distinguishing the correlations that drive GDOP inflation from those that can be arranged away.

\vspace{-10pt}

\subsection{The Seven-State and Eight-State Formulations}\label{subsec:formulations}

To analyze how the clock bias column couples to the other navigation states, we separate the Jacobian into a 7-state matrix -- position, velocity, and clock drift -- and the full 8-state matrix that restores the clock bias column. Define the 7-state Jacobian $H_7 = \begin{bmatrix} A_p & A_v & A_d \end{bmatrix} \in \mathbb{R}^{N \times 7}$, obtained by removing the clock bias block $A_b$ from Eq.~\eqref{eq:jacobian}, with columns ordered as position, velocity, and clock drift. The 8-state Jacobian appends the clock bias column in the final position, $H_8 = \begin{bmatrix} H_7 & \mathbf{c} \end{bmatrix} \in \mathbb{R}^{N \times 8}$, where $\mathbf{c} = A_b = [c^1, \ldots, c^N]^T$ is the vector of clock bias sensitivities from Eq.~\eqref{eq:coupling_analytical}. The matrix $H_8$ contains the same columns as $A_{\mathrm{GDOP}}$ in Eq.~\eqref{eq:jacobian} with the clock bias column reordered to the last position; the reordering leaves the GDOP unchanged and places clock bias as the bordered row and column of the Gram matrix below. The corresponding Gram matrices are $M_7 = H_7^T H_7$ and $M_8 = H_8^T H_8$, with covariance matrices $G_k = M_k^{-1}$ ($k = 7, 8$). The block structure of $M_8$ is
\begin{equation}\label{eq:M8_block}
M_8 = \begin{bmatrix} M_7 & \mathbf{b} \\ \mathbf{b}^T & d \end{bmatrix}
\end{equation}
where $\mathbf{b} = H_7^T \mathbf{c} \in \mathbb{R}^7$ contains the inner products between $\mathbf{c}$ and each column of $H_7$, and $d = \|\mathbf{c}\|^2 = \sum_{j=1}^N (c^j)^2$. The vector $\mathbf{b}$ partitions as $\mathbf{b} = [\mathbf{b}_p^T,\; \mathbf{b}_v^T,\; b_d]^T$, where $\mathbf{b}_p = A_p^T \mathbf{c} \in \mathbb{R}^3$, $\mathbf{b}_v = A_v^T \mathbf{c} \in \mathbb{R}^3$, and $b_d = A_d^T \mathbf{c} = \sum_{j=1}^N c^j$ quantify the coupling of the clock bias column to the position, velocity, and drift subspaces, respectively.

\vspace{-10pt}

\subsection{Structure of the Clock-Bias Coupling}\label{subsec:coupling_structure}

The sensitivity element $c^j$ from Eq.~\eqref{eq:coupling_column} comprises the two projections of Eqs.~\eqref{eq:term1} and~\eqref{eq:term2}; writing Term~2 through the normalized LOS-rate magnitude $\tilde S^j$ of Eq.~\eqref{eq:S_scaling}, they read
\begin{equation}\label{eq:term1_restate}
\underbrace{\frac{(\hat{\vec\rho}^j)^T \dot{\bar{\vec{v}}}^j}{\eta^j}}_{\text{Term 1}} = \frac{1-\beta^j}{\beta^j}\,(\beta^j z^j + \kappa^j),
\qquad
\underbrace{\frac{(\dot{\hat{\vec\rho}}^j)^T \bar{\vec{v}}^j}{\eta^j}}_{\text{Term 2}} = -\frac{\kappa^j}{\beta^j(1-\beta^j)}\,(\tilde S^j)^2,
\end{equation}
whose sum is the general sensitivity $c^j = c(z^j;\beta^j,\chi^j)$ of Eq.~\eqref{eq:coupling_analytical}, of which the in-plane reading $c_{\mathrm{ip}}^j = -z^j\tilde S(z^j;\beta^j,0)$ in Eq.~\eqref{eq:c_eq_zStilde} is the $\chi=0$ value. Term~1 is the projection of the relative acceleration onto the LOS direction and is aspect-angle-independent; Term~2 is the projection of the relative velocity onto the LOS-rate direction and carries the aspect angle through $\tilde S^j$. For $\beta = 0.891$ (780~km) at zenith, $|\text{Term 2}| = 1/\beta \approx 1.12$ and $|\text{Term 1}| = (1-\beta)/\beta \approx 0.12$, so Term~2 dominates by the factor $1/(1-\beta) \approx 9$; the aspect-angle term vanishes at zenith, so these values hold at every aspect angle. Toward the horizon the two terms cancel only when the LOS lies in the orbital plane; off-plane, $|\text{Term 2}|$ exceeds $|\text{Term 1}|$ and the sensitivity persists. The behavior of Term~2 has a kinematic origin: by Eq.~\eqref{eq:rho_dot_fundamental}, $\dot{\hat{\vec\rho}}^j = \vec{v}_{\mathrm{sat},\perp}^j/\rho^j$ for a stationary receiver, so that $(\dot{\hat{\vec\rho}}^j)^T \bar{\vec{v}}^j = -(\dot{\hat{\vec\rho}}^j)^T \vec{v}_{\mathrm{sat}}^j = -\rho^j\|\dot{\hat{\vec\rho}}^j\|^2$, giving $\text{Term 2} = -\rho^j\|\dot{\hat{\vec\rho}}^j\|^2/\eta^j$. Because $\|\dot{\hat{\vec\rho}}^j\| = \gamma^j\tilde S^j$, Term~2 grows with $(\tilde S^j)^2$, the squared magnitude of the position block. The clock-bias sensitivity is, therefore, tied to the LOS-rate magnitude that also sets the position sensitivity; satellites with larger position sensitivity carry larger $|c^j|$, creating a structural correlation between the clock bias column and the position columns.

Separately, the sensitivity $c^j = c(z^j;\beta^j,\chi^j)$ of Eq.~\eqref{eq:coupling_analytical} carries the elevation factor $z^j$, the third (Up) component of the LOS unit vector (Eq.~\eqref{eq:los_enu}). Because the velocity block $A_v$ has rows $(\hat{\vec\rho}^j)^T$ (Eq.~\eqref{eq:8state_row}), its third column -- the velocity $z$-column -- is the stack of these third LOS components across the visible satellites, so its $j$th entry is exactly $\hat{\vec\rho}_z^j = z^j$. The clock bias column and the velocity $z$-column, therefore, share the same elevation factor $z^j$: at fixed altitude and aspect angle this is a functional relationship, and across a constellation it confines the two columns to the band between the envelopes of Eq.~\eqref{eq:coupling_envelopes}. The aspect-angle-independent Term~1 enters only through $(z,\beta)$ and has no azimuthal dependence, so any correlation between $c^j$ and the horizontal velocity components arises only from the specific constellation geometry, not from the measurement physics.

Both couplings just described act on top of a more basic property: the sign of $c^j$ does not change with geometry. Because $c^j = -\ddot{\rho}^j/\eta^j$, the bound $|c^j| \le 1$ in each satellite's own normalization is inherent in the construction of $\eta^j$ as an upper bound on the range acceleration \parencite{Psiaki2021}; what the closed form adds is the exact statement: strict negativity for every satellite above the horizon, equivalently the strict positivity of the range acceleration, attainment of the bound only at zenith, and the dependence on the aspect angle between the two envelopes. The next lemma collects these properties, which, combined with Term~2 above, force the clock bias column to align with the seven-state subspace and underlie the irreducible drift correlation established below.

\begin{lemma}[Sign and Bound of Clock-Bias Sensitivity]\label{lem:c_negative}
For all LEO altitudes ($\beta \in (0,1)$) and all aspect angles $\chi$, the clock-bias sensitivity of Eq.~\eqref{eq:coupling_analytical} satisfies $c(z;\beta,\chi) < 0$ for every satellite above the horizon ($z \in (0,1]$), and $c(0;\beta,\chi) = -(1-\beta)\beta\sin^2\chi/\kappa \le 0$ at the horizon, vanishing only for a LOS in the orbital plane ($\chi = 0$). At zenith, $c(1;\beta,\chi) = -1$ at every aspect angle. In particular $|c(z;\beta,\chi)| \le 1$, with equality only at zenith.
\end{lemma}

\begin{proof}
In Eq.~\eqref{eq:coupling_analytical}, for $\beta \in (0,1)$ the factor $1-\beta > 0$ and $\kappa > 0$; the bracket is a sum of nonnegative terms, with $z(\beta z+\kappa) > 0$ for $z > 0$ because $\beta z+\kappa > 0$. Hence $c < 0$ for $z \in (0,1]$. At $z = 0$ the first term vanishes and $\cos^2\varepsilon = 1$, giving $c = -(1-\beta)\beta\sin^2\chi/\kappa$, which is zero only when $\chi = 0$.

The magnitude is largest at $\chi = \pi/2$, where $c$ equals the cross-track envelope $c_{\mathrm{xt}}(z;\beta) = -(1-\beta)(z\kappa+\beta)/\kappa$ of Eq.~\eqref{eq:coupling_envelopes}, so it suffices to bound $|c_{\mathrm{xt}}|$. The inequality $|c_{\mathrm{xt}}| \le 1$ is equivalent to $(1-\beta)(z\kappa+\beta) \le \kappa$, that is,
\begin{equation}\label{eq:c_bound_reduction}
\kappa\bigl[\,1 - (1-\beta)z\,\bigr] \;\ge\; \beta(1-\beta) .
\end{equation}
Both factors on the left are positive and decreasing in $z$ on $[0,1]$: the bracket is affine in $z$ with value $1$ at $z = 0$ and $\beta$ at $z = 1$, and $\mathrm{d}\kappa/\mathrm{d}z = \beta\bigl(\beta z/\sqrt{\beta^2 z^2 + 1 - \beta^2} - 1\bigr) < 0$ because $\beta z < \sqrt{\beta^2 z^2 + 1 - \beta^2}$. Their product is therefore decreasing in $z$, with minimum at $z = 1$, where $\kappa = 1-\beta$ and the left-hand side equals $(1-\beta)\beta$. Hence \eqref{eq:c_bound_reduction} holds on $z \in [0,1]$, with equality only at $z = 1$. Because $|c(z;\beta,\chi)| \le |c_{\mathrm{xt}}(z;\beta)|$, it follows that $|c| \le 1$, with equality only at zenith, where $\kappa = 1-\beta$ and $c = -1$.
\end{proof}

At zenith, the sensitivity is $c = -1$ at every aspect angle, the most negative value attainable; this is independent of altitude because $\eta$ is the zenith range acceleration. At the horizon the sensitivity runs across the band $[\,c_{\mathrm{xt}}(0;\beta),\,0\,]$ as the aspect angle varies from cross-track to in-plane, so an in-plane horizon-grazing satellite produces no first-order clock-bias signature in the Doppler measurement while a cross-track one retains a finite sensitivity. This floor is also where the stationary, non-rotating-Earth model matters: receiver motion at the percent level of the orbital speed can lift a near-floor entry to zero or marginally above, so the strict per-satellite sign holds within the stated model, whereas the aggregate results below depend only on the negativity of the sum, which a near-zero entry does not affect. Every above-horizon satellite contributes $c^j \in [-1,0)$, so for $N$ satellites the total clock-bias information $d = \|\mathbf{c}\|^2 = \sum_j (c^j)^2$ is at most $N$, attained when all satellites are at zenith. Because $c^j < 0$ for every above-horizon satellite, the drift inner product is nonzero whenever at least one satellite is strictly above the horizon, as stated next.

\begin{corollary}[Irreducible Clock-Drift Correlation]\label{cor:drift_correlation}
For any constellation containing at least one satellite strictly above the horizon, $b_d = \sum_{j=1}^N c^j < 0$. The clock bias column is always negatively correlated with the clock-drift column, independently of the geometric arrangement of visible satellites.
\end{corollary}

The vector $\mathbf{b} = H_7^T\mathbf{c} = [\mathbf{b}_p^T,\; \mathbf{b}_v^T,\; b_d]^T$ has components with distinct structural properties. The drift coupling $b_d = \sum c^j < 0$ is irreducible (Corollary~\ref{cor:drift_correlation}). The velocity $z$-component coupling $b_{v,z} = \sum z^j c^j$ is sign-definite: every term is nonpositive, because $z^j \ge 0$ and $c^j \le 0$ by Lemma~\ref{lem:c_negative}, and the sum is strictly negative for any constellation with a satellite above the horizon, so it cannot be arranged away. The position coupling $\mathbf{b}_p = A_p^T\mathbf{c}$ is structurally nonzero because Term~2 of each $c^j$ grows with $(\tilde S^j)^2$, the squared magnitude of the position columns, aligning $\mathbf{c}$ with those columns. By contrast, the horizontal velocity couplings $b_{v,x} = \sum x^j c^j$ and $b_{v,y} = \sum y^j c^j$, where $x^j$ and $y^j$ denote the horizontal components of $(\hat{\vec\rho}^j)^T$, can vanish for suitably symmetric constellations, because $c^j$ carries no direct azimuthal dependence. The irreducible couplings -- drift and the vertical velocity component -- are the ones that drive the collinearity between the clock bias column and the seven-state subspace toward unity at fixed altitude.

\vspace{-10pt}

\subsection{The Observability--Correlation Tradeoff}\label{subsec:obs_corr}

The results above reveal a fundamental tradeoff. Clock bias is observable from Doppler measurements only through its effect on the computed range rate, and by Section~\ref{subsec:coupling_structure} the dominant sensitivity term satisfies $\text{Term 2} \propto \|\dot{\hat{\vec\rho}}\|^2$, the same LOS-rate quantity that sets the position sensitivity. This shared dependence is not a design flaw but the mechanism by which clock bias becomes observable: without it, a clock-bias error would produce no Doppler signature. The same coupling, however, forces $\mathbf{c}$ to project onto $\mathrm{col}(H_7)$, so the covariance is always inflated, by an amount set by geometry and quantified exactly in the next section. The effect has a pseudorange analogue: there the clock bias column of ones correlates with the vertical LOS component, coupling clock bias to altitude error, whereas here $c^j$ correlates with the vertical LOS-\emph{rate} component, shifting that coupling from the position-altitude domain to the velocity domain.

\vspace{-10pt}

\section{Schur Complement Bounds}\label{sec:schur}


This section quantifies exactly how much adding the clock bias column inflates the estimation covariance. Using the block structure of the eight-state Gram matrix, we apply the Schur complement to obtain a closed-form expression for the GDOP increase in terms of two quantities: the collinearity coefficient, which measures how much of the clock bias column lies in the span of the seven-state subspace, and the total clock bias sensitivity. The added covariance is nonnegative by the extra-state property; we further show that perfect decorrelation is impossible for any constellation, and that the inflation admits a geometric interpretation as an orthogonal projection. These bounds set up the altitude-diversity mechanism of Section~\ref{sec:altitude}, which acts directly on the collinearity coefficient.

\vspace{-10pt}

\subsection{Schur Complement and Covariance Inflation}\label{subsec:schur_complement}

The block structure of the Gram matrix in Eq.~\eqref{eq:M8_block} admits exact inversion via the Schur complement \parencite{BarShalom2001}. Throughout this section the design matrices are taken to be of full column rank, $\operatorname{rank}(H_7) = 7$ and $\operatorname{rank}(H_8) = 8$, which requires at least $N \geq 8$ visible satellites and renders the Gram matrices $M_7$ and $M_8$ invertible and the eight-state GDOP finite; the selection problem of interest, drawing a subset from the tens to hundreds of visible LEO satellites, is comfortably overdetermined. Define $S \equiv d - \mathbf{b}^T G_7 \mathbf{b}$. Under this assumption $S > 0$, and the covariance matrix $G_8 = M_8^{-1}$ (cf.~Section \ref{subsec:formulations}) takes the form
\begin{equation}\label{eq:G8_schur}
G_8 = \begin{bmatrix} G_7 + \frac{1}{S} G_7 \mathbf{b}\mathbf{b}^T G_7 & -\frac{1}{S} G_7 \mathbf{b} \\[0.5em] -\frac{1}{S} \mathbf{b}^T G_7 & \frac{1}{S} \end{bmatrix}
\end{equation}
The Schur complement of a partitioned information matrix is the classical route to isolating the effect of an added parameter on the covariance of the others; block partitioning of the dilution-of-precision problem has a multi-constellation precedent in the closed-form GDOP formula of \textcite{TengWang2016}, and what is specific here is its application to a geometry-dependent added column, the clock bias column. The Schur complement used here is the block-elimination operation on the information matrix, distinct from the Schur-convexity of the determinant invoked in the volume-heuristic literature of Section~\ref{subsec:pr_gdop}; the two share a name but refer to different objects. The GDOP of the $k$-state problem is
\begin{equation}\label{eq:gdop_def}
\mathrm{GDOP}_k = \sqrt{\mathrm{Tr}\,G_k}, \qquad k = 7, 8,
\end{equation}
the square root of the trace of the corresponding covariance matrix. Let $G_8^{(7)}$ denote the leading $7\times7$ principal submatrix of $G_8$, the covariance of the seven original states in the presence of the clock bias column. Comparing $G_8^{(7)}$ with $G_7$ yields the covariance inflation $G_8^{(7)} = G_7 + \Delta G$, where the inflation matrix $\Delta G = \frac{1}{S} G_7 \mathbf{b}\mathbf{b}^T G_7$. The structure of $\Delta G$ determines whether estimating clock bias can ever improve the precision of the other seven states. The next theorem shows it cannot: the inflation term is positive semidefinite, so adding the clock bias column raises the variance of every original state or leaves it unchanged, never lowering it. This monotonicity is the extra-state property that estimating an additional parameter cannot reduce the covariance of the others, established for the least-squares problem by \textcite{vandiggelen2009agps} and applied to GNSS dilution of precision by \textcite{BlancoDelgado2010}; the contribution here is the exact closed form of the inflation matrix $\Delta G$ and its decomposition by state group.
\begin{theorem}[Extra-State Monotonicity]\label{thm:covariance_inflation}
Adding the clock bias column does not decrease the covariance of any original state: $G_8^{(7)} \succeq G_7$, where $\succeq$ denotes the positive semidefinite ordering.
\end{theorem}

\begin{proof}
The inflation matrix $\Delta G = \frac{1}{S} G_7 \mathbf{b}\mathbf{b}^T G_7$ is a rank-1 outer product scaled by $1/S > 0$. For any $\mathbf{x} \in \mathbb{R}^7$, $\mathbf{x}^T \Delta G\,\mathbf{x} = (\mathbf{b}^T G_7 \mathbf{x})^2/S \geq 0$, so $\Delta G \succeq 0$. The inequality is strict in any direction not orthogonal to $G_7\mathbf{b}$.
\end{proof}

The total GDOP inflation follows from the trace:
\begin{equation}\label{eq:gdop_increase}
\mathrm{GDOP}_8^2 - \mathrm{GDOP}_7^2 = \frac{1}{S}\left(\|\boldsymbol{\xi}\|^2 + 1\right) > 0,
\qquad \boldsymbol{\xi} \equiv G_7\mathbf{b},
\end{equation}
where $\boldsymbol{\xi} = G_7\mathbf{b}$ is the coefficient vector that expresses the projection of $\mathbf{c}$ onto $\mathrm{col}(H_7)$ in the column basis, and the $+1$ accounts for the clock bias variance itself (the $(8,8)$ entry of $G_8$). Thus, $\mathrm{GDOP}_8 > \mathrm{GDOP}_7$ whenever $H_8$ has full column rank.

\vspace{-10pt}

\subsection{Geometric Interpretation and the Collinearity Coefficient}\label{subsec:collinearity}

The Schur complement admits a geometric interpretation that clarifies its role.

\begin{theorem}[Projection Interpretation]\label{thm:S_projection}
The Schur complement equals the squared norm of the component of $\mathbf{c}$ orthogonal to the column space of $H_7$, $S = \|P_{H_7^\perp}\,\mathbf{c}\|^2$, where $P_{H_7^\perp} = I - H_7(H_7^TH_7)^{-1}H_7^T$ is the orthogonal projector onto $\mathrm{col}(H_7)^\perp$.
\end{theorem}

\begin{proof}
From $\mathbf{b} = H_7^T\mathbf{c}$ and $G_7 = (H_7^TH_7)^{-1}$, we have $\mathbf{b}^TG_7\mathbf{b} = \mathbf{c}^T H_7(H_7^TH_7)^{-1}H_7^T\mathbf{c} = \mathbf{c}^TP_{H_7}\,\mathbf{c} = \|P_{H_7}\,\mathbf{c}\|^2$, where $P_{H_7} = H_7(H_7^TH_7)^{-1}H_7^T$ is the orthogonal projector onto $\mathrm{col}(H_7)$. Because $P_{H_7} + P_{H_7^\perp} = I$, it follows that $S = \|\mathbf{c}\|^2 - \|P_{H_7}\,\mathbf{c}\|^2 = \|P_{H_7^\perp}\,\mathbf{c}\|^2$.
\end{proof}

The decomposition carries over unchanged to weighted measurements: with a positive-definite weight $W$, the inverse of the measurement-noise covariance, replacing each Gram matrix by its $W$-weighted form leaves the inflation identity intact, the projection of Theorem~\ref{thm:S_projection} being taken in the $W$-weighted inner product; the unweighted form is retained here because the analysis targets the geometric mechanism rather than a noise-dependent error budget.

Define the collinearity coefficient as the fraction of $\mathbf{c}$ lying in the column space of $H_7$,
\begin{equation}\label{eq:rho_def}
\zeta^2 \equiv \frac{\|P_{H_7}\,\mathbf{c}\|^2}{\|\mathbf{c}\|^2} = \frac{\mathbf{b}^TG_7\mathbf{b}}{d} \in [0,1]
\end{equation}
The Schur complement then takes the form
\begin{equation}\label{eq:S_rho}
S = d(1 - \zeta^2)
\end{equation}
and the GDOP inflation can be written compactly as
\begin{equation}\label{eq:gdop_inflation_final}
\mathrm{GDOP}_8^2 - \mathrm{GDOP}_7^2 = \frac{\|\boldsymbol{\xi}\|^2 + 1}{d(1 - \zeta^2)}
\end{equation}
where $\boldsymbol{\xi} = G_7\mathbf{b}$ is the least-squares coefficient vector introduced above, which best approximates $\mathbf{c}$ as a linear combination of the columns of $H_7$. The numerator captures how well the clock bias column aligns with sensitive state directions; the denominator measures how much of $\mathbf{c}$ is independent of the 7-state subspace.

The inflation decomposes by state group. Partitioning $\boldsymbol{\xi}$ conformally as $\boldsymbol{\xi} = [\boldsymbol{\xi}_p^T,\; \boldsymbol{\xi}_v^T,\; \xi_d]^T$ (position, velocity, drift), the covariance increase for each block is
\begin{equation}
[G_8]_{pp} - [G_7]_{pp} = \frac{1}{S}\,\boldsymbol{\xi}_p\boldsymbol{\xi}_p^T \in \mathbb{R}^{3 \times 3}, \qquad
[G_8]_{vv} - [G_7]_{vv} = \frac{1}{S}\,\boldsymbol{\xi}_v\boldsymbol{\xi}_v^T \in \mathbb{R}^{3 \times 3}, \qquad
[G_8]_{dd} - [G_7]_{dd} = \frac{1}{S}\,\xi_d^2 \in \mathbb{R}
\label{eq:dop_inflation_all}
\end{equation}
Each inflation term is a rank-1 outer product scaled by $1/S$: the variance of each state component either increases or remains unchanged, consistent with Theorem~\ref{thm:covariance_inflation}. A large $\|\boldsymbol{\xi}_p\|$ indicates that the clock bias column is well-approximated by the position columns, causing position estimation to suffer the most; similarly for $\|\boldsymbol{\xi}_v\|$ (velocity) and $|\xi_d|$ (drift). Which channel dominates depends on the constellation geometry and is evaluated for a representative configuration in Section~\ref{sec:example}. Figure~\ref{fig:schur_flow} summarizes the complete decomposition, from the block information matrix through the collinearity coefficient to the GDOP inflation.

\vspace{-10pt}

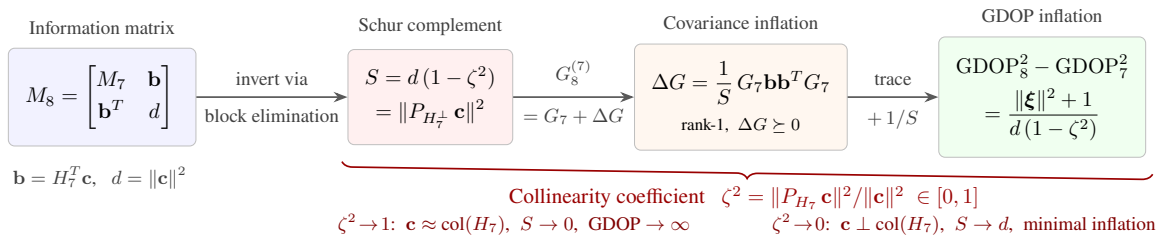
\begin{figure}[htb]
\centering
\scalebox{0.8}{%
\begin{tikzpicture}[
    >=Stealth,
    ann/.style={font=\footnotesize},
    stepbox/.style={draw=gray!50, rounded corners=3pt, inner sep=8pt,
                    fill=white, font=\normalsize, align=center},
    arr/.style={->, thick, gray!70!black},
    annlabel/.style={font=\small, text=gray!50!black, align=center},
  ]
 
  \node[stepbox, fill=blue!5] (M8) {%
    $M_8 = \begin{bmatrix}
      M_7 & \mathbf{b} \\[3pt]
      \mathbf{b}^T & d
    \end{bmatrix}$
  };
  \node[annlabel, above=3pt] at (M8.north) {Information matrix};
 
  \node[annlabel, below=6pt] at (M8.south) {%
    $\mathbf{b} = H_7^T\mathbf{c}$, \;
    $d = \|\mathbf{c}\|^2$
  };
 
  \node[stepbox, fill=red!6, right=2.5cm of M8] (schur) {%
    $S = d\,(1 - \zeta^2)$\\[4pt]
    ${\displaystyle = \|P_{H_7^\perp}\,\mathbf{c}\|^2}$
  };
  \node[annlabel, above=3pt] at (schur.north) {Schur complement};
 
  \draw[arr] (M8) -- (schur)
    node[midway, above=2pt, annlabel] {invert via}
    node[midway, below=2pt, annlabel] {block elimination};
 
  \node[stepbox, fill=orange!6, right=2.0cm of schur] (covinf) {%
    $\Delta G = \dfrac{1}{S}\,G_7\mathbf{b}\mathbf{b}^TG_7$\\[6pt]
    {\footnotesize rank-1, \;$\Delta G \succeq 0$}
  };
  \node[annlabel, above=3pt] at (covinf.north) {Covariance inflation};
 
  \draw[arr] (schur) -- (covinf)
    node[midway, above=2pt, annlabel] {$G_8^{(7)}$}
    node[midway, below=2pt, annlabel] {$= G_7 + \Delta G$};
 
  \node[stepbox, fill=green!6, right=1.5cm of covinf] (gdop) {%
    $\mathrm{GDOP}_8^2 - \mathrm{GDOP}_7^2$\\[4pt]
    ${\displaystyle = \frac{\|\boldsymbol{\xi}\|^2 + 1}
      {d\,(1 - \zeta^2)}}$
  };
  \node[annlabel, above=3pt] at (gdop.north) {GDOP inflation};
 
  \draw[arr] (covinf) -- (gdop)
    node[midway, above=2pt, annlabel] {trace}
    node[midway, below=2pt, annlabel] {$+\,1/S$};
 
  \coordinate (rhoL) at ($(schur.south west)+(-0.1, -0.8)$);
  \coordinate (rhoR) at ($(gdop.south east)+(0.1, -0.8)$);
  \coordinate (rhoMid) at ($(rhoL)!0.5!(rhoR)$);
 
  \draw[decorate, decoration={brace, amplitude=5pt, mirror, raise=-15pt},
        red!60!black, thick]
    (rhoL) -- (rhoR);
 
  \node[below=-10pt, font=\normalsize, red!60!black, align=center] at (rhoMid) {%
    Collinearity coefficient \;
    $\zeta^2 = \|P_{H_7}\,\mathbf{c}\|^2 / \|\mathbf{c}\|^2
    \;\in [0,1]$
  };
 
  \node[annlabel, red!50!black, align=center]
    at ($(rhoMid)+(0, -0.4)$) {%
    $\zeta^2 \!\to\! 1$:\; $\mathbf{c} \approx$
    col$(H_7)$, \;$S \to 0$,
    \;GDOP $\to \infty$
    \hspace{1.2cm}
    $\zeta^2 \!\to\! 0$:\; $\mathbf{c} \perp$
    col$(H_7)$, \;$S \to d$,
    \;minimal inflation};
 
\end{tikzpicture}
}%
\caption{Schur complement decomposition of the eight-state information matrix. The collinearity coefficient $\zeta^2$ measures the fraction of the clock bias column $\mathbf{c}$ lying within the column space of $H_7$; the residual $1-\zeta^2$ determines the Schur complement $S$ and, through Eq.~\eqref{eq:gdop_inflation_final}, the GDOP inflation.}
\label{fig:schur_flow}
\end{figure}

\vspace{-10pt}

\subsection{Bounds and Limiting Cases}\label{subsec:bounds}

The coupling structure established in Section~\ref{sec:coupling} imposes constraints on the achievable collinearity.

\begin{corollary}[Impossibility of Perfect Decorrelation]\label{cor:high_collinearity}
For any constellation with at least one satellite strictly above the horizon, $\zeta^2 > 0$ at any altitude. Perfect decorrelation ($\zeta^2 = 0$) is impossible.
\end{corollary}

\begin{proof}
By Theorem~\ref{thm:S_projection}, $\zeta^2 = 0$ requires $\mathbf{b} = H_7^T\mathbf{c} = \mathbf{0}$. However, $b_d = \sum_{j=1}^N c^j < 0$ by Corollary~\ref{cor:drift_correlation}, so $\mathbf{b} \neq \mathbf{0}$ and $\zeta^2 > 0$.
\end{proof}

Corollary~\ref{cor:high_collinearity} excludes the lower endpoint $\zeta^2 = 0$; the opposite endpoint is attained exactly. When the visible satellites share a common elevation, altitude, and aspect angle, the clock-bias entries collapse onto a single value, and the column becomes a scalar multiple of the clock-drift column already contained in $H_7$.

\begin{proposition}[Single-Shell Collinearity Ceiling]\label{prop:collinearity_ceiling}
If every visible satellite shares a common elevation, altitude, and aspect angle, $z^j = z$, $\beta^j = \beta$, and $\chi^j = \chi$ for all $j$, then the clock bias column is a scalar multiple of the clock-drift column, $\mathbf{c} = c(z;\beta,\chi)\,\mathbf{1}$, so that $\zeta^2 = 1$ and $S = 0$, and the eight-state GDOP is unbounded.
\end{proposition}

\begin{proof}
Under the stated conditions every entry $c^j = c(z;\beta,\chi)$ of Eq.~\eqref{eq:coupling_analytical} takes the same value, so $\mathbf{c} = c(z;\beta,\chi)\,\mathbf{1}$. The clock-drift column $A_d = \mathbf{1}$ is a column of $H_7$, hence $\mathbf{c} \in \mathrm{col}(H_7)$ and $P_{H_7}\mathbf{c} = \mathbf{c}$. By Eq.~\eqref{eq:rho_def}, $\zeta^2 = \|P_{H_7}\mathbf{c}\|^2/\|\mathbf{c}\|^2 = 1$, and by Eq.~\eqref{eq:S_rho}, $S = d(1-\zeta^2) = 0$.
\end{proof}

The two results bracket the coefficient, $\zeta^2 \in (0,1]$, with the floor never reached and the ceiling reached precisely when the satellites are indistinguishable in elevation, altitude, and aspect angle, the degenerate case in which the full-rank assumption fails and the eight-state dilution is unbounded. Decorrelation is the separation of the satellites along these three variables, developed in Section~\ref{sec:altitude}. A competing constraint limits how far a useful constellation departs from the ceiling: the eight-state GDOP rewards a large total sensitivity $d = \|\mathbf{c}\|^2$, largest for the high-elevation, low-altitude satellites that lie nearest the ceiling, where $|c| \to 1$ and the clock-bias and velocity $z$-columns are most nearly proportional. The geometry that minimizes GDOP is therefore drawn toward high collinearity; configurations with small $\zeta^2$ exist but carry small $d$ and poor GDOP, so a near-optimal constellation settles at a high but sub-unity collinearity, as the worked example of Section~\ref{sec:example} confirms. The clock bias variance admits simple bounds in terms of $d = \|\mathbf{c}\|^2$.

\begin{proposition}[Clock Bias Variance Lower Bound]\label{prop:clock_bias_bounds}
For $N$ satellites with clock bias sensitivities $c^1, \ldots, c^N$,
\begin{equation}
\sigma_{\delta_R}^2 \;\geq\; \frac{1}{\sum_{j=1}^N (c^j)^2} \;=\; \frac{1}{d}
\end{equation}
\end{proposition}

\begin{proof}
From Eq.~\eqref{eq:G8_schur}, $\sigma_{\delta_R}^2 = 1/S = 1/[d(1-\zeta^2)]$. Because $H_8$ has full column rank, $S > 0$ and $0 \leq \zeta^2 < 1$, so $0 < S \leq d$, giving $\sigma_{\delta_R}^2 \geq 1/d$. The lower bound is achieved when $\zeta^2 = 0$, which is never exactly attained.
\end{proof}

Defining the pairwise correlation between the clock bias column and the velocity $z$-column as
\begin{equation}\label{eq:rcz_def}
r_{c,z} \equiv \frac{\sum_{j=1}^N z^j\,c^j}{\sqrt{\sum_{j=1}^N (z^j)^2}\;\sqrt{d}}
\end{equation}
the collinearity coefficient $\zeta^2$ measures how well $\mathbf{c}$ is predicted by all seven state columns jointly, whereas $r_{c,z}^2$ captures only its correlation with a single velocity component. By construction, $\zeta^2 \geq r_{c,z}^2$. At fixed altitude, the velocity $z$-component correlation dominates, and the two quantities are numerically close ($\zeta^2 \approx r_{c,z}^2$); with altitude diversity, other column correlations contribute, and the gap widens. The Schur complement $S = d(1-\zeta^2)$ remains the proper measure of GDOP inflation throughout; $r_{c,z}$ is useful as a diagnostic because it isolates the single largest source of collinearity.

Equation~\eqref{eq:gdop_inflation_final} explicates the two competing effects that govern Doppler GDOP: the collinearity $\zeta^2$, which drives the denominator toward zero as the clock bias column becomes more aligned with the 7-state subspace, and the total sensitivity $d = \|\mathbf{c}\|^2$, which counteracts the inflation by increasing $S$. When the collinearity is mild, the inflation admits a simple bound independent of the Schur complement, so the 7-state geometry serves as a usable proxy for the full 8-state problem. The next corollary makes this precise: once the clock bias column is more orthogonal than parallel to the 7-state subspace, the inflation is bounded in terms of $\|\boldsymbol{\xi}\|^2$ and $d$, although $\|\boldsymbol{\xi}\|$ itself remains a geometry-dependent quantity. Whether this condition is attainable is the question taken up in Table~\ref{tab:elevation_bounds} and Section~\ref{sec:altitude}.
\begin{corollary}[Sufficient Condition for 7-State Proxy]\label{cor:proxy_sufficient}
If $\zeta^2 < 1/2$ (clock bias column is more orthogonal than parallel to the 7-state subspace), then
\begin{equation}\label{eq:proxy_bound}
\mathrm{GDOP}_8^2 - \mathrm{GDOP}_7^2 < \frac{2(\|\boldsymbol{\xi}\|^2 + 1)}{d}
\end{equation}
\end{corollary}

\begin{proof}
From Eq.~\eqref{eq:gdop_inflation_final} with $S = d(1-\zeta^2)$, $\zeta^2 < 1/2$ gives $S > d/2$, so $1/S < 2/d$.
\end{proof}

Table~\ref{tab:elevation_bounds} illustrates, at fixed altitude, how $d$ and $|r_{c,z}|$ respond to the elevation distribution and to the aspect angle. The per-satellite sensitivity spans the band between its in-plane and cross-track envelopes, so $d$ itself ranges within the band shown; high-elevation satellites increase $d$ but drive $|r_{c,z}|$ toward unity, while low-elevation satellites raise $d$ only off-plane. When every satellite shares one elevation and one aspect angle, the clock-bias and velocity $z$-columns are proportional and $|r_{c,z}| = 1$ exactly, so the aspect angle does not decorrelate an equal-elevation configuration. The all-horizon row is the degenerate exception: the clock bias column vanishes in-plane, and although it is nonzero off-plane, $|r_{c,z}|$ is undefined because the velocity $z$-column vanishes at zero elevation. Across the configurations of Table~\ref{tab:elevation_bounds}, $|r_{c,z}|$ stays near unity at fixed altitude: a large $d$ is produced by high-elevation satellites, which are precisely those that drive $|r_{c,z}|$ toward unity, so at fixed altitude the total sensitivity and the decorrelation cannot be raised together; breaking this tension requires altitude diversity (Section~\ref{sec:altitude}).

\vspace{-10pt}

\begin{table}[htbp!]
\caption{Sensitivity band for $N = 8$, $\beta = 0.891$ (780~km), evaluated at the in-plane floor ($\chi = 0$) and cross-track ceiling ($\chi = \pi/2$) of Eq.~\eqref{eq:coupling_envelopes}. The band on $d$ shows how the aspect angle shifts the sensitivity magnitude at fixed altitude, while $|r_{c,z}|$ stays near unity throughout.}
\label{tab:elevation_bounds}
\centering
\footnotesize
\begin{tabular}{lcccc}
 & \multicolumn{2}{c}{$d$} & \multicolumn{2}{c}{$|r_{c,z}|$} \\
\cmidrule(lr){2-3}\cmidrule(lr){4-5}
Configuration & in-plane & cross-track & in-plane & cross-track \\
\midrule
All zenith ($z^j = 1$)            & 8.00 & 8.00 & 1.00 & 1.00 \\
All horizon ($z^j = 0$)           & 0    & 0.37 & --   & --   \\
4 zenith + 4 horizon              & 4.00 & 4.18 & 1.00 & 0.98 \\
Uniform ($z^j = 0.5$)             & 0.26 & 2.55 & 1.00 & 1.00 \\
Linear spread$^\ddagger$          & 1.01 & 2.98 & 0.95 & 0.99 \\
\bottomrule
\multicolumn{5}{l}{\footnotesize $^\ddagger$Linear spread: $z^j = j/9$ for $j = 1,\ldots,8$.}\\
\end{tabular}
\end{table}

\vspace{-5pt}

\section{Altitude Diversity and the Sensitivity--Decorrelation Tradeoff}\label{sec:altitude}
This section identifies altitude diversity as the primary mechanism that breaks the collinearity established in Section~\ref{sec:coupling}, and characterizes its limits. We first show that satellites at different altitudes but the same elevation carry different clock bias sensitivities, displacing them from the single-shell band that forces high collinearity and thereby increasing the Schur complement; the satellite heading provides a second, weaker lever within that band. We then identify the competing effect: higher-altitude satellites decorrelate the clock bias column but carry weaker Doppler signatures, reducing the total sensitivity, so the optimal allocation balances the two. The two effects favor extreme altitude separation over graduated spacing for the satellites that carry appreciable coupling, predicting a bang-bang tendency for those satellites, while satellites near zenith are altitude-degenerate; the optimization of Section~\ref{sec:example} bears this out across optimizer restarts.

\vspace{-10pt}

\subsection{Altitude Diversity as a Decorrelation Mechanism}\label{subsec:decorrelation_mech}

At a fixed orbital altitude, the clock bias sensitivity $c^j = c(z^j;\beta^j,\chi^j)$ is confined to the band between the in-plane and cross-track envelopes, which depend on elevation and altitude alone (Proposition~\ref{prop:latitude_independence}), so for any single-altitude constellation every satellite lies within the same single-shell band. Across the high-sensitivity elevation distributions that the GDOP objective favors, this confinement drives $\zeta^2$ toward the single-shell ceiling of Proposition~\ref{prop:collinearity_ceiling}, as Section~\ref{subsec:bounds} established; the aspect angle moves each satellite only within the band, a weak lever because the band is narrow relative to the elevation factor $\hat{\vec\rho}_z^j = z^j$ shared with the velocity $z$-column.

Altitude diversity breaks this alignment. Two satellites at the same elevation $\varepsilon$ but different altitudes ($\beta_1 \neq \beta_2$) share the velocity $z$-column entry $\hat{\vec\rho}_z = z$ but carry different clock-bias sensitivities, because the sensitivity of Eq.~\eqref{eq:coupling_analytical} varies with $\beta$ at fixed elevation and heading. This displaces satellites off the single-shell band, reducing $\zeta^2$ and increasing the Schur complement $S = d(1-\zeta^2)$. The variation with altitude is governed by $\partial c/\partial\beta$ at fixed elevation and heading, which vanishes at zenith, where $c = -1$ at every altitude and heading, and is appreciable for satellites away from zenith. Zenith satellites are therefore altitude-degenerate: their clock-bias sensitivity does not change with altitude, and any decorrelation between altitudes enters only through the $\eta^j/\eta_{\mathrm{ref}}$ scaling. The intrinsic separation between altitudes is greatest for the coupling-carrying satellites away from zenith. The quantitative impact on $\zeta^2$ is evaluated for a representative configuration in Section~\ref{sec:example}.

\vspace{-10pt}

\subsection{The Sensitivity--Decorrelation Tradeoff}\label{subsec:sd_tradeoff}

The scaling parameters $\gamma$ and $\eta$ decrease sharply with increasing orbital altitude (Table~\ref{tab:scaling_altitude}), reflecting the reduced Doppler signatures available from higher orbits. Across the LEO range of 400--1200~km, $\gamma$ varies by a factor of approximately $3.2$, and $\eta$ by a factor of $3.75$ (Table~\ref{tab:scaling_altitude}). Under common reference normalization, a low-altitude satellite contributes position columns roughly twice as large as one at the reference altitude, representing a genuine per-measurement information advantage.

This creates a fundamental tension. Low-altitude satellites provide stronger per-measurement sensitivity (larger $d = \|\mathbf{c}\|^2$) but are confined to a single-shell band, producing high $\zeta^2$. High-altitude satellites offer decorrelation through a different $\beta$ value but reduce $d$ through weaker sensitivity. The inflation formula in Eq.~\eqref{eq:gdop_inflation_final} contains both $d$ (in the denominator) and $\zeta^2$ (in the factor $1-\zeta^2$), so the optimal altitude allocation must balance sensitivity loss against collinearity reduction.

The tradeoff has diminishing returns. Adding a high-altitude satellite replaces a low-altitude one, decreasing $d$ (because $|c^j|$ is smaller at higher altitude) while reducing $\zeta^2$ (because the altitude contrast displaces the satellite from the single-shell band). The product $S = d(1-\zeta^2)$ increases only if the decorrelation gain outweighs the sensitivity loss. The first high-altitude satellite provides the largest decorrelation benefit, because it introduces a second band; subsequent high-altitude satellites offer diminishing marginal decorrelation while incurring the same per-satellite sensitivity cost ($\gamma^{400}/\gamma^{1200} \approx 3.2$). This asymmetry predicts that optimal configurations will allocate a small minority of satellites to high altitude, a prediction confirmed in the next section.

\vspace{-10pt}

\begin{table}[htb]
\caption{Scaling Parameter Variation with Orbital Altitude}
\label{tab:scaling_altitude}
\centering
\footnotesize
\begin{tabular}{lcccc}
\toprule
Altitude (km) & $\beta$ & $\gamma$ (mrad/s) & $\eta$ (m/s$^2$) & $\gamma/\gamma_{780}$ \\
\midrule
400 (ISS)      & 0.941 & 19.2  & 138.4 & 2.01 \\
550 (Starlink) & 0.921 & 13.8  &  96.3 & 1.44 \\
780 (Iridium)  & 0.891 &  9.57 &  63.67 & 1.00 \\
1200 (OneWeb)  & 0.842 &  6.05 &  36.92 & 0.63 \\
\bottomrule
\end{tabular}%
\vspace{-10pt} 
\end{table}

\vspace{-5pt}

\subsection{Implications for Constellation Geometry}\label{subsec:implications}

Both tendencies driving altitude selection favor extremes rather than graduated spacing. Sensitivity increases monotonically as altitude decreases (Table~\ref{tab:scaling_altitude}), so reducing a satellite's altitude improves its measurement quality throughout the LEO range. Decorrelation favors maximizing $|\beta_1 - \beta_2|$ to widen the separation between bands, which again favors the widest possible altitude separation. These two tendencies act on the satellites that carry appreciable coupling, namely those away from zenith, where $\partial c/\partial\beta$ is appreciable and altitude contrast actively decorrelates the clock bias column; for such a satellite both tendencies push its altitude toward an extreme. A satellite at or near zenith behaves differently: it lies on a flat direction of the objective in altitude, because $c = -1$ at every altitude and heading there and $\partial c/\partial\beta \to 0$, so its altitude is weakly constrained and the decorrelation it provides enters only through the $\eta^j/\eta_{\mathrm{ref}}$ scaling. This structure -- a bang-bang tendency toward the altitude extremes for the coupling-carrying satellites and altitude-degeneracy for those near zenith -- leaves the split between co-optimal configurations only weakly determined, and is examined across a family of optimizer restarts in Section~\ref{sec:example}.

However, the altitude separation available in current multi-constellation environments is modest. Between Starlink (${\sim}550$~km, $\beta = 0.921$) and OneWeb (1200~km, $\beta = 0.842$), $\Delta\beta \approx 0.08$. At this separation, the sensitivity cost of diverting satellites from low to high altitude limits the achievable $\zeta^2$ reduction, and the narrow tradeoff window constrains the effectiveness of altitude-based heuristics. Figure~\ref{fig:cz_scatter} illustrates the decorrelation mechanism.

\vspace{-10pt}

\begin{figure}[htb]
\centering
\includegraphics[width=0.75\textwidth]{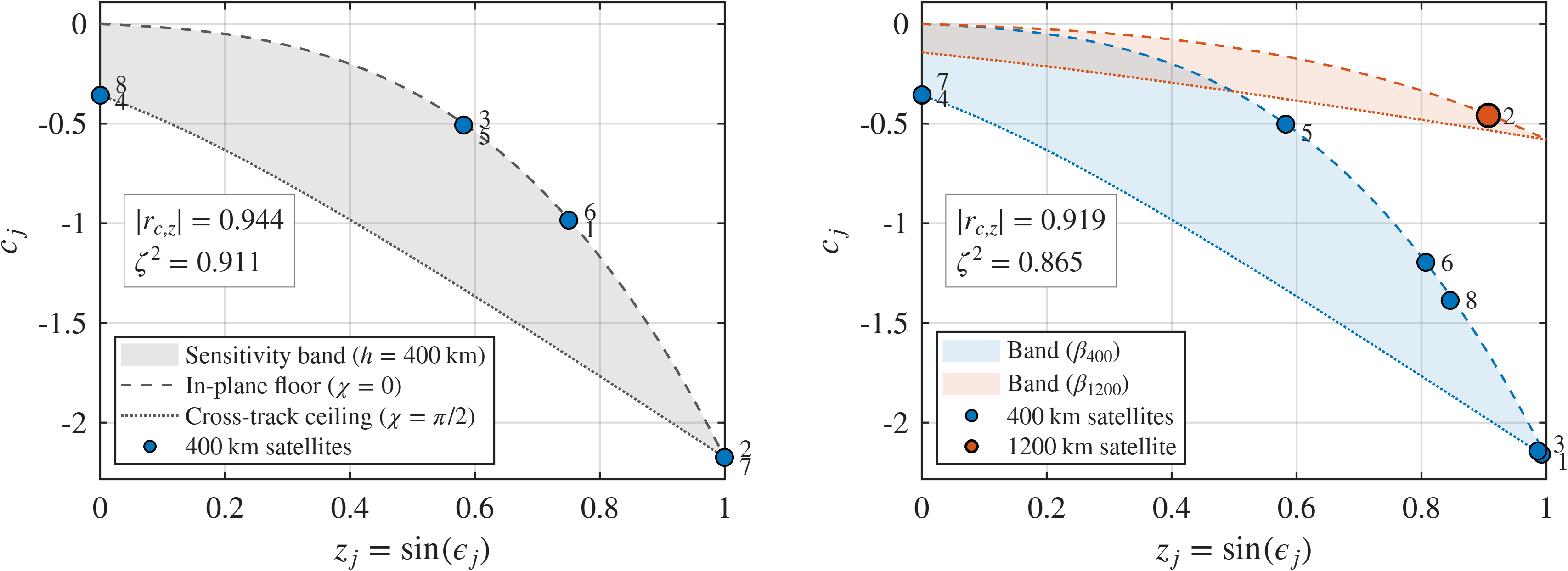}
\caption{Illustration of the decorrelation mechanism. Left: at a single altitude, satellites lie within the narrow band between the in-plane and cross-track envelopes $c(z;\beta)$, forcing high collinearity between the clock bias column and the velocity $z$-column. Right: with altitude diversity, satellites at a different orbital altitude occupy a distinct band $c(z;\beta')$, reducing $\zeta^2$ and increasing the Schur complement $S$. Entries and envelopes are shown in a common normalization at a 780~km reference altitude, so magnitudes exceed the per-satellite bound of Lemma~\ref{lem:c_negative}.} 
\label{fig:cz_scatter}
\vspace{-10pt} 
\end{figure}

\vspace{-10pt}

\section{Example: Multi-Altitude Single Epoch}\label{sec:example}
This section illustrates the analysis of Sections~\ref{sec:schur} and~\ref{sec:altitude} on a representative configuration, separating the effect of altitude level from that of altitude diversity, verifying the exact GDOP inflation formula, and confirming the structure of the optimized solutions across optimizer restarts. We optimize a synthetic two-shell constellation and compare its Schur complement diagnostics against two single-shell baselines, one at a reference altitude and one at the low altitude to which the optimizer drives the constellation, so that the second comparison holds altitude level fixed and isolates the contribution of altitude diversity. We then confirm the inflation formula to machine precision at the optimized geometry.

\vspace{-10pt}

\subsection{Setup}\label{subsec:example_setup}

To illustrate the predictions of Sections~\ref{sec:schur} and~\ref{sec:altitude}, we consider a synthetic two-shell constellation with $N = 8$ satellites whose altitudes are drawn from the range $[400,\,1200]$~km, together with two single-shell baselines used as controls. The Jacobian is normalized to a common reference altitude rather than per satellite, as \textcite{Psiaki2021} prescribes for constellations spanning several orbital radii, because a per-satellite normalization would rescale away the altitude dependence of the Doppler sensitivity that altitude diversity exploits, leaving the geometric mechanism invisible; a common reference preserves the relative sensitivity of satellites at different altitudes. The reference is set to 780~km ($\gamma_{\mathrm{ref}} = 9.57$~mrad/s, $\eta_{\mathrm{ref}} = 63.67$~m/s$^2$), the midpoint of the $[400,\,1200]$~km bracket rounded to the Iridium shell; the choice fixes only the overall scale of the dilution values, which cancels in the comparisons because every configuration shares the normalization.

An unconstrained elevation mask ($0^\circ$) is used so that the comparison isolates the geometric mechanism from operational visibility constraints, which depend on receiver and mission specifics. The two single-shell baselines place all eight satellites at a common altitude, one at the 780~km reference and one at 400~km, the low shell to which the optimizer drives the bulk of the variable-altitude constellation, so that the second comparison varies only the altitude spread.

Satellite positions and altitudes are determined by numerically minimizing the 8-state GDOP of Eq.~\eqref{eq:gdop_def} over elevation, azimuth, inclination, pass direction, and altitude, with the single-shell baselines fixing the altitude. These variables are not independent for a physical orbit: at each evaluation the satellite latitude implied by the sky position constrains the admissible inclination through $\min(i^j, \pi - i^j) \geq |\phi_{sat}^j|$, and Clairaut's relation \parencite{Geyer2016} then supplies the ground-track heading from the inclination, the pass direction, and that satellite latitude. The aspect angle and the apparent heading follow from the ground-track heading and the sky position. Every candidate, therefore, corresponds to a realizable circular-orbit state. These feasibility constraints are evaluated for a receiver on the equator. By Proposition~\ref{prop:latitude_independence} the sensitivity band is independent of this choice, whereas the aspect angles realized within it are not, so the optimized geometry reported below is specific to the equatorial receiver. The structural results of Sections~\ref{sec:coupling} and~\ref{sec:schur} hold independently of these feasibility constraints, which enter only here. The configuration is, nonetheless, an instantaneous Jacobian rather than a propagated constellation: two satellites may occupy the same sky direction at a single epoch -- for instance near the zenith -- which no single constellation realizes simultaneously, just as the regular tetrahedron that minimizes pseudorange GDOP is an instantaneous optimum rather than a sustained arrangement. The example therefore illustrates the column structure that governs Doppler GDOP, and is not advanced as a constellation-design prescription.

The minimization uses particle swarm optimization \parencite{Kennedy1995}, a population-based stochastic search; the swarm comprises 400 particles run for up to 500 iterations, and to reduce sensitivity to local minima the optimizer is run from 30 independent random initializations, the lowest-GDOP solution being retained. The eight-state optimum is stable across these restarts, reproducing to a coefficient of variation of 1.5\%. The decision-variable bounds, the feasibility map, the restart protocol, and the condition numbers of $M_7$ and $M_8$ at the retained optima were recorded; no ill-conditioned configuration was retained.

\subsection{Numerical Evaluation}\label{subsec:example_eval}

The optimizer converges to a two-shell structure that places all satellites at the altitude extremes: in the retained solution, seven satellites at 400~km and one at 1200~km, with none at intermediate altitudes. The single high-shell satellite is a coupling-carrying satellite at mid-to-high elevation ($\varepsilon \approx 65^\circ$), where both $|c|$ and its altitude sensitivity $\partial c/\partial\beta$ are appreciable; two of the low-shell satellites sit near the horizon and carry the smallest sensitivity, and the remaining low-shell satellites span mid-to-high elevations. The structure is robust across restarts: in all 30 converged restarts every satellite settles at either the lower or the upper altitude bound and none at an intermediate altitude, consistent with the monotone-in-altitude tendency of Section~\ref{subsec:implications}. The split between the two shells is not sharply determined -- a seven-plus-one arrangement appears in most restarts and a six-plus-two arrangement in the remainder, the two reaching the same eight-state GDOP to within the 1.5\% restart variability -- so the shell count is co-optimal across a small family of configurations rather than uniquely fixed. Table~\ref{tab:example_results} compares the Schur complement diagnostics for this variable-altitude configuration against the two single-shell baselines, and Table~\ref{tab:example_config} lists the optimized variable-altitude geometry in full. Throughout the comparison $\mathrm{GDOP}_7$ is evaluated at the same geometry that minimizes $\mathrm{GDOP}_8$, so that the two dilution values share a common configuration and their difference isolates the cost of adding the clock bias column; this is the quantity the inflation formula of Eq.~\eqref{eq:gdop_inflation_final} predicts. The geometry that minimizes $\mathrm{GDOP}_7$ in its own right attains a lower value still, but comparing across two different geometries would conflate the inflation with the change of operating point.

\begin{table}[htb]
\caption{Schur Complement Diagnostics: Single-Shell Baselines and Variable Altitude ($N = 8$, $0^\circ$ Mask)}
\label{tab:example_results}
\centering
\footnotesize
\begin{tabular}{lccc}
\toprule
Quantity & Single (780~km) & Single (400~km) & Variable (400 + 1200~km) \\
\hline
$d = \|\mathbf{c}\|^2$            & 2.87             & 12.17            & 13.31 \\
$\zeta^2$                          & 0.918            & 0.911            & 0.865 \\
$S = d(1-\zeta^2)$                 & 0.234            & 1.084            & 1.793 \\
$(1-\zeta^2)^{-1}$                 & 12.27            & 11.23            & 7.42  \\
$\mathrm{GDOP}_7$                 & 2.46$^{\dagger}$ & 1.92$^{\dagger}$ & 1.89$^{\dagger}$ \\
$\mathrm{GDOP}_8$                 & 3.51             & 2.53             & 2.25  \\
\bottomrule
\multicolumn{4}{l}{\footnotesize $^\dagger$$\mathrm{GDOP}_7$ evaluated at the corresponding 8-state optimal geometry.}
\end{tabular}
\vspace{-10pt} 
\end{table}

\begin{table}[htb]
\caption{Optimized variable-altitude geometry ($N = 8$), sorted by elevation; the configuration of Figure~\ref{fig:example_geometry_varalt}. Here $c$ is the clock-bias entry in the common normalization, and the pass direction is ascending (A) or descending (D).}
\label{tab:example_config}
\centering
\footnotesize
\setlength{\tabcolsep}{4pt}
\begin{tabular}{cccccccc}
\toprule
Sat & Alt (km) & $\varepsilon$ ($^\circ$) & Az ($^\circ$) & $i$ ($^\circ$) & Pass & $\chi$ ($^\circ$) & $c$ \\
\midrule
4 &  400 &  0.0 &  89.0 &  89.7 & A & 88.8 & $-0.357$ \\
7 &  400 &  0.0 & 249.8 &  74.2 & D & 86.8 & $-0.356$ \\
5 &  400 & 35.7 & 351.8 &  81.3 & D &  0.6 & $-0.503$ \\
6 &  400 & 53.8 & 349.2 & 100.0 & A &  0.8 & $-1.197$ \\
8 &  400 & 57.8 & 182.6 & 100.0 & A & 12.6 & $-1.387$ \\
2 & 1200 & 65.1 & 156.7 &  52.6 & D & 14.2 & $-0.458$ \\
3 &  400 & 80.5 & 118.8 & 100.0 & D & 71.3 & $-2.141$ \\
1 &  400 & 82.8 & 343.4 &   4.4 & A & 77.8 & $-2.158$ \\
\bottomrule
\end{tabular}
\vspace{-10pt} 
\end{table}

The comparison separates two effects that the single-shell baselines hold apart. Lowering the constellation from the 780~km reference to a single shell at 400~km raises the total sensitivity $d$ from 2.87 to 12.17, a factor of $4.2$, driven by the per-satellite advantage $(\eta^{400}/\eta_{\mathrm{ref}})^2 \approx 4.7$: each 400~km satellite carries roughly $4.7$ times the clock-bias information of a 780~km satellite at the same elevation. This is an altitude-\emph{level} effect that any low constellation enjoys, and it leaves the collinearity essentially unchanged, $\zeta^2 = 0.918$ against $0.911$, so the Schur complement rises almost in proportion to $d$, from 0.234 to 1.084, and $\mathrm{GDOP}_8$ falls from 3.51 to 2.53.

Altitude \emph{diversity} is the distinct effect isolated by comparing the 400~km single shell with the variable-altitude constellation, which share the same low shell for seven of their eight satellites. Raising one satellite to 1200~km lowers the collinearity from $\zeta^2 = 0.911$ to $0.865$ as it separates from the single-shell band (Figure~\ref{fig:cz_scatter}), so that $(1-\zeta^2)^{-1}$ falls from 11.23 to 7.42; the Schur complement rises from 1.084 to 1.793 and $\mathrm{GDOP}_8$ improves from 2.53 to 2.25, a further reduction of about ten percent. Of the rise in the Schur complement from 0.234 at the reference shell to 1.793 with diversity, the step to 1.084 at the 400~km single shell is the altitude-level contribution and the step from 1.084 to 1.793 is the diversity contribution, factors of $4.6$ and $1.65$ respectively, so the diversity-specific gain is the smaller of the two by a wide margin. Within the single shell, the satellite heading is a weaker lever still, lowering $\zeta^2$ by $0.024$ from its in-plane value of $0.935$, against the $0.046$ removed by altitude diversity.

At the reference altitude the eight-state optimum carries a markedly higher dilution than at 400~km, with $\mathrm{GDOP}_8 = 3.51$ against $\mathrm{GDOP}_7 = 2.46$ evaluated at the same geometry; this gap is the covariance cost of the clock-bias coupling, reduced both by flying lower and by spreading in altitude. Figure~\ref{fig:example_geometry_varalt} contrasts the 8-state optimum at the 400~km single shell with the variable-altitude optimum, showing that the coupling-carrying satellite at mid-to-high elevation is the one drawn to the high-altitude shell, while the near-horizon satellites, which carry the smallest sensitivity, and the near-zenith satellites, whose sensitivity is altitude-degenerate, remain at low altitude; the satellite the optimizer raised is the one where altitude contrast most reduces the collinearity.

\begin{figure}[htb]
\centering
\includegraphics[width=0.6\textwidth]{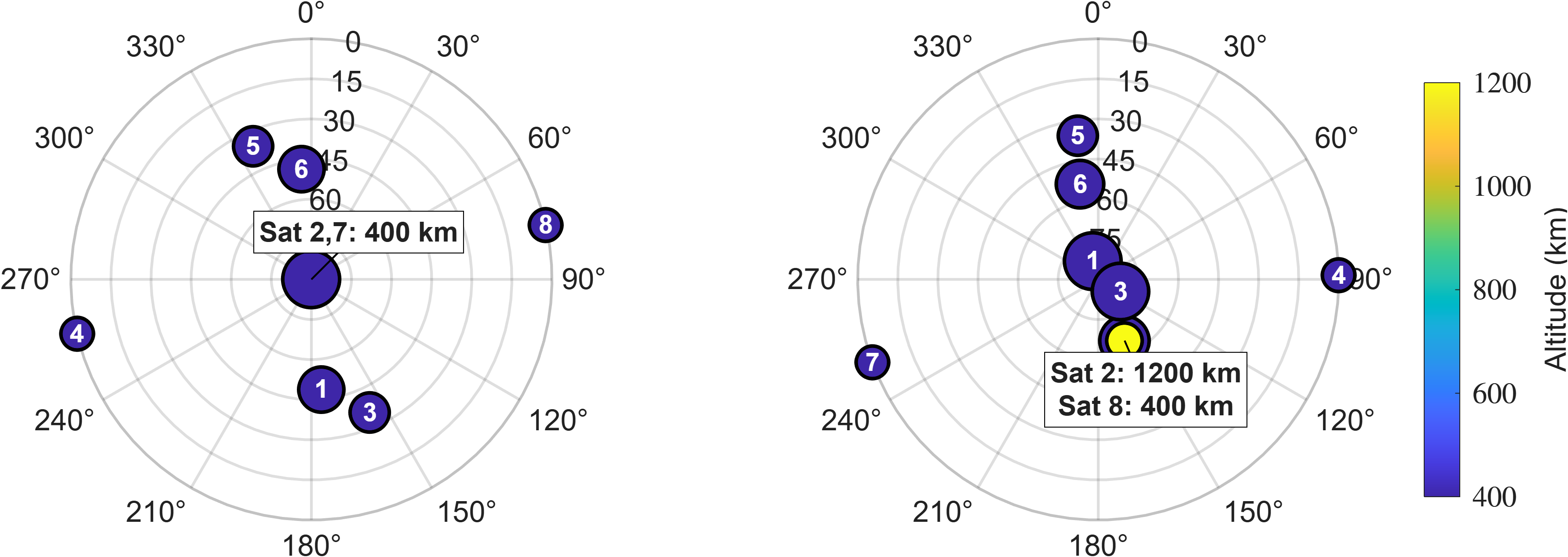}
\caption{Eight-state optimal satellite geometries: the separately optimized 400~km single-shell baseline (left) and the variable-altitude configuration tabulated in Table~\ref{tab:example_config} (right, seven satellites at 400~km and one at 1200~km). Satellite indices are assigned per configuration and do not correspond across panels. Radial distance from center represents zenith angle ($90^\circ - \varepsilon$); marker color encodes orbital altitude for the variable-altitude configuration. The single coupling-carrying satellite drawn to the high-altitude shell illustrates the decorrelation mechanism of Section~\ref{sec:altitude} for one representative case rather than as a general design outcome.}
\label{fig:example_geometry_varalt}
\vspace{-10pt}
\end{figure}

The clock-bias entries in Tables~\ref{tab:example_results} and~\ref{tab:example_config} and Figures~\ref{fig:cz_scatter} and~\ref{fig:example_geometry_varalt} exceed unity because they are expressed in the common normalization rather than each satellite's own: the per-satellite bound $|c|\le1$ of Lemma~\ref{lem:c_negative} is rescaled by $\eta^j/\eta_{\mathrm{ref}} \approx 2.2$ for a 400~km satellite, so its entry reaches about $2.2$ near zenith -- the same sensitivity amplification that drives the increase in $d$.

\vspace{-10pt}

\subsection{Verification of Theoretical Predictions}\label{subsec:example_verify}

The GDOP inflation formula in Eq.~\eqref{eq:gdop_inflation_final} is verified at the variable-altitude optimal geometry. Evaluating both $\mathrm{GDOP}_8$ and $\mathrm{GDOP}_7$ at the same satellite configuration yields $\mathrm{GDOP}_8^2 - \mathrm{GDOP}_7^2 = 5.05 - 3.56 = 1.49$, with $\mathrm{GDOP}_8 = 2.247$ and $\mathrm{GDOP}_7 = 1.886$ at the shared geometry. The Schur complement prediction gives $(\|\boldsymbol{\xi}\|^2 + 1)/S = (1.671 + 1)/1.793 = 1.49$, where $\boldsymbol{\xi} = G_7\mathbf{b}$ is the least-squares coefficient vector of Eq.~\eqref{eq:gdop_inflation_final} with $\|\boldsymbol{\xi}\|^2 = 1.671$; the two agree to within $10^{-14}$, confirming the identity. As an independent check that the value of $\mathrm{GDOP}_8$ itself is correct, and not merely algebraically self-consistent, the same figure is obtained by directly evaluating $\sqrt{\operatorname{Tr}[(A_{\mathrm{GDOP}}^T A_{\mathrm{GDOP}})^{-1}]}$ on the unpermuted eight-state Jacobian, which reproduces $\mathrm{GDOP}_8 = 2.247$. The inflation $\mathrm{GDOP}_8^2 - \mathrm{GDOP}_7^2$ falls from 6.27 at the reference shell to 2.70 at 400~km and 1.49 with diversity, the larger step again accompanying the altitude-level drop, which lifts the variance floor $1/d$ (Proposition~\ref{prop:clock_bias_bounds}) from 0.35 to 0.082 and leaves little room for inflation. As a check on the parameterization itself, the closed-form rows of Eq.~\eqref{eq:8state_row} were compared against Jacobians obtained by finite differencing the measurement model on Cartesian satellite states, with the clock-bias sensitivity obtained by perturbing the state-evaluation epoch, over a grid of elevations, azimuths, inclinations, pass directions, and altitudes, agreeing to a maximum normalized error of $1\times10^{-8}$.

\vspace{-10pt}

\section{Discussion and Conclusions}\label{sec:conclusions}

\vspace{-5pt}

This paper developed an analytical framework for understanding Doppler-based navigation geometry with LEO satellites. Starting from the closed-form parameterization of the eight-state Jacobian, we showed that the clock bias column is structurally correlated with both the position columns, through the shared LOS rate that sets the measurement sensitivity, and the vertical velocity column, through the shared elevation factor. The Schur complement decomposition quantified the resulting GDOP inflation exactly in terms of the collinearity coefficient and the total clock bias sensitivity, and established that the inflation is always strictly positive: the same geometric coupling that makes clock bias observable from Doppler measurements is the coupling that inflates estimation covariance. The size of that inflation is not fixed but set by geometry, which is what makes altitude diversity able to reduce it.

The central theoretical finding is the sensitivity--decorrelation tradeoff. Altitude diversity reduces the collinearity by separating satellites at equal elevation onto distinct sensitivity bands, and the satellite heading provides a weaker reduction within each band, but higher-altitude satellites carry weaker Doppler signatures. Both levers are modest. A single-altitude constellation that is favorable for Doppler GDOP is already strongly collinear, because the same high-sensitivity satellites that lower the dilution are the ones least separable from the clock-drift direction, and at the altitude separations available in current LEO constellations the relief that altitude diversity provides is limited: in the worked configuration most of the geometric improvement comes from flying the constellation lower, an altitude-level effect available to any low constellation, while the part attributable to altitude diversity itself is a small fraction of that total.

These results give a structural account of an effect observed empirically in prior work: that the volume of the LOS vectors is a poor predictor of Doppler GDOP while a determinant over the full Doppler Jacobian is a strong one. The analysis here identifies why, by locating the clock bias column outside the span that a LOS volume can represent and showing that its per-satellite sensitivity is sign-definite and bounded, and its coupling to the remaining states irreducible. The same analysis bounds how far geometry can act on that coupling: altitude diversity lowers the collinearity but the available altitude spread is small, and the satellite heading is weaker still, so neither offers more than partial relief. The same structure answers the question of whether a cheaper geometric proxy could stand in for the full Jacobian: a seven-state proxy is licensed only when the clock bias column is more orthogonal than parallel to the remaining states, a condition far from attainable in any configuration examined here, so the reliance of satellite selection on the full Doppler Jacobian is structural rather than incidental. The treatment is kept dimensionless throughout, because its object is the geometric mechanism that governs the relative quality of constellations rather than an absolute error budget, which depends on the measurement-noise model and the operational scenario.

\vspace{-10pt}

\section*{Acknowledgments}

This work was supported by the Peter Munk Research Institute (PMRI) at Technion and  the Israeli Smart Transportation Research Center.

\vspace{-10pt}

\section*{Conflict of Interest}
The authors declare no conflict of interest.

\vspace{-10pt}

\printbibliography[title=References]

\end{document}